\documentclass[a4paper,12pt,twoside,reqno]{amsart}

\usepackage[T1]{fontenc}
\usepackage[utf8]{inputenc}
\usepackage[english]{babel}

\usepackage[osf]{mathpazo}
\usepackage{eucal}
\usepackage{mathrsfs}
\usepackage{microtype}

\usepackage{amsmath}
\usepackage{amssymb}
\usepackage{amsthm}
\usepackage{mathtools}
\usepackage{upgreek}
\usepackage{enumerate}

\usepackage{graphicx}
\usepackage{framed}
\usepackage{url}
\usepackage{float}
\usepackage{tikz}
\usetikzlibrary{decorations.pathreplacing}
\usetikzlibrary{decorations.markings}
\usetikzlibrary{shapes.geometric}
\usetikzlibrary{matrix}
\usetikzlibrary{patterns,calc}
\usepackage{circuitikz}
\usepackage{tikz-cd}
\usepackage{pgfplots}
\pgfplotsset{compat=1.17,colormap/blackwhite}
\usepackage{caption}
\usepackage{subcaption}
\usepackage{wrapfig}
\usepackage[shortlabels]{enumitem}
\usepackage{xcolor}
\usepackage{relsize}
\usepackage{setspace}
\allowdisplaybreaks

\DeclarePairedDelimiter\ket{\lvert}{\rangle}
\DeclarePairedDelimiterX\braket[2]{\langle}{\rangle}{#1\,\delimsize\vert\,\mathopen{}#2}
\usepackage{braket}

\usepackage{hyperref}

\makeatletter
\@ifundefined{CMcal}{\DeclareMathAlphabet{\CMcal}{OMS}{cmsy}{m}{n}}{}
\makeatother

\theoremstyle{plain}
\newtheorem{teo}{Theorem}[section]
\newtheorem{cor}[teo]{Corollary}
\newtheorem{prop}[teo]{Proposition}
\newtheorem{lemma}[teo]{Lemma}

\theoremstyle{definition}
\newtheorem{definition}[teo]{Definition}

\theoremstyle{remark}
\newtheorem{remark}[teo]{Remark}

\numberwithin{equation}{section}
\numberwithin{figure}{section}
\numberwithin{table}{section}

\newcommand{\eprint}[1]{\href{https://arxiv.org/abs/#1}{arXiv:#1}}
\newcommand{\cites}[1]{\cite{#1}}

\newcommand\inner[2]{\left\langle #1, #2 \right\rangle}

\newcommand{\CC}{\mathbb{C}}

\newcommand{\NN}{\mathbb{N}}

\newcommand{\ZZ}{\mathbb{Z}}

\newcommand{\pa}[1]{\left( #1 \right)}
\def\norm#1{\left\|#1\right\|}

\newcommand{\f}[1]{\mathcal{#1}}

\newcommand{\n}[1]{\mathbb{#1}}

\newcommand{\ab}[1]{\left| #1 \right|}

\hypersetup{
	colorlinks=true,
	linkcolor=purple,
	citecolor=teal,
	urlcolor=blue,
	filecolor=magenta,
	breaklinks=true,
	bookmarksopen=true,
	pdftitle={Higher Abelian Quantum Double Models},
	pdfauthor={Jorge Acu\~{n}a Flores, Giuseppe De Nittis, Javier Lorca Espiro},
}

\begin{document}
	
	\title[Higher Abelian Quantum Double Models]{Higher Abelian Quantum Double Models}
	
	\author[J. Acu\~na Flores]{Jorge Acu\~{n}a Flores}
	\address{Facultad de Matem\'{a}ticas, Pontificia Universidad Cat\'{o}lica de Chile, Santiago, Chile}
	\email{coke.acuna@uc.cl}
	
	\author[G. De~Nittis]{Giuseppe De Nittis}
	\address{Facultad de Matem\'{a}ticas \& Instituto de F\'{i}sica, Pontificia Universidad Cat\'{o}lica de Chile, Santiago, Chile}
	\email{gidenittis@uc.cl}
	
	\author[J. Lorca Espiro]{Javier Lorca Espiro}
	\address{Departamento de Ciencias F\'{i}sicas, Facultad de Ingenier\'{i}a y Ciencias, Universidad de La Frontera, Avda. Francisco Salazar 01145, Casilla 54-D Temuco, Chile}
	\email{javier.lorca@ufrontera.cl}
	
	\date{\today}
	
	\begin{abstract}
		This paper develops a rigorous $C^*$-algebraic framework for higher abelian quantum double models, generalizing Kitaev's construction to simplicial complexes of arbitrary finite dimension and local regularity. We fully characterize the frustration-free ground state space through an associated algebra of logical operators: its state space is shown to be homeomorphic to the space of frustration-free ground states, and to obey generalized canonical commutation relations. When the relevant homology and cohomology groups are finite, the logical algebra decomposes into a commutative factor and a full matrix algebra, thereby separating the classical and quantum parts of the frustration-free ground state structure. We further prove that the vanishing of these groups is necessary and sufficient for the core algebra of the model to form a Cartan pair with the full algebra of observables, a property expected to pave the way toward classifying the model's equilibrium (KMS) states, extending recent results for the planar case.
	\end{abstract}
	
	\maketitle
	\tableofcontents

	\section{Introduction}

	Topological phases of matter have emerged as a central paradigm in modern condensed matter physics, offering a rich landscape of phenomena that go beyond the traditional Landau-Ginzburg framework of spontaneous symmetry breaking. These phases are characterized by robust topological properties, such as ground-state degeneracy dependent on the system's geometry and the presence of anyonic excitations with non-trivial braiding statistics (see, for instance, \cite{2008PhRvL.101a0504L}, \cite{2010PhRvL.104m0502F}, and references therein). Among the most prominent models exhibiting topological order is Kitaev's toric code \cite{Kitaev_2003}, a two-dimensional lattice system with spins residing on the links and a Hamiltonian composed of mutually commuting \textit{star} and \textit{plaquette} operators. When defined on the torus, the toric code has a fourfold degenerate ground state manifold, protected against local perturbations by a spectral gap, with excitations given by abelian anyons; the same construction can be carried out on surfaces of other topologies -- the name ``toric code'' is nonetheless retained -- and Kitaev already observed in \cite{Kitaev_2003} that on a closed orientable surface of genus $g$ the ground space degeneracy equals $4^g$, directly tying the ground state degeneracy to the topology of the underlying manifold. While these anyons are not sufficient for universal quantum computation, the model provides a paradigmatic example of topological protection and anyonic statistics, and its simplicity and exact solvability make it a valuable tool for exploring topological quantum memory.
	
	\medskip
	
	The toric code has a well-known family of generalizations called quantum double models (QDMs), already introduced by Kitaev himself in \cite{Kitaev_2003} and further studied in \cites{Naaijkens-2011,Cui_2020,Bols_2025}. These extend the framework from spin systems to discrete gauge theories with finite gauge groups and support anyonic excitations that can be used for fault-tolerant quantum computation. However, standard QDMs usually require planar qubit layouts and exhibit a vanishing encoding rate as the code distance increases, limiting their scalability: as the linear size of the lattice grows while its topology is kept fixed, the dimension of the ground state (code) space stays essentially constant, so the ratio between the number of encoded logical qubits and the number of physical qubits tends to zero. Improving on this behaviour requires letting the relevant topological data -- and hence the ground state degeneracy -- grow together with the system size, which is part of what motivates the higher-dimensional and higher-gauge generalizations discussed below.
	
	\medskip
	
	Going beyond the restrictions of surface codes is a significant theoretical challenge. One promising approach involves models with generalized symmetries, particularly \textit{higher gauge models}. Among them, models constructed via homological and cohomological techniques \cite{dealmeida2017topological} preserve key topological features of QDMs, such as ground state degeneracy and entanglement properties, while allowing definitions on arbitrary compact $n$-dimensional manifolds. In the $N$-dimensional abelian case, \cite{Vrana} shows that the structure of frustration-free ground states is governed by the homology and locally finite cohomology of the underlying CW-complex, forming a logical algebra capable of encoding classical and quantum information, and that localized excitations are classified by homology and cohomology at infinity. A crucial difference with Vrana’s work is that we incorporate relations between different gauge levels (higher gauge structure), which prevents the decomposition of the Hamiltonian into independent terms per dimension. This necessitates the use of Brown homology and cohomology, as in \cite{dealmeida2017topological}, to rigorously describe the ground state space. In the case of a compact manifold, the classification of the ground state subspace has been carried out from a more geometrical approach in \cite{Lorca2025}, while a rigorous understanding of its relation to the excited states has yet to be established. The dependence of the ground state degeneracy of abelian quantum double models on the topology of the underlying closed orientable manifold has also been established, with different techniques, by Bachmann \cite{Bachmann2017}; our results can be seen as extending this type of statement to the higher-gauge, non-compact setting.
	
	\medskip
	
	The algebra of operators we construct to capture the ground state degeneracy -- what we call the $C^*$-algebra of logical operators $\mathfrak{A}_{\rm log}$, see Section \ref{sec:analysis} -- has a formal parallel with the annular logical algebras used in the study of two-dimensional topological order, although the objects serve different purposes. For local commuting-projector Hamiltonians, Haah \cite{Haah2016} considers operators supported on an annulus that commute with the relevant local Hamiltonian terms, modulo an ideal of null operators; the resulting quotient algebra encodes particle types and is used to construct an invariant related to the topological \(S\)-matrix. Kato and Naaijkens \cite{KatoNaaijkens2020} build on this annular logical-algebra structure to define an entropic invariant determined by its superselection-sector data. By contrast, our algebra
	\[
	\mathfrak{A}_{\rm log}=\mathfrak{K}'/\mathfrak{J}
	\]
	is not attached to a distinguished annular region and is not introduced to classify localized excitations. It is a global quotient of the stabilizer commutant whose state space parametrizes all frustration-free ground states, and its generators and commutation relations are determined intrinsically by the compactly supported homology and cohomology of the underlying geometric and algebraic chain complexes. The common structural principle is the passage from operators compatible with the local constraints to a quotient that removes locally trivial or physically redundant operators. Establishing a precise relation between these constructions, even in cases where both apply, such as the planar toric code, remains an interesting problem for future work.
	
	\medskip
	
	In this paper, we provide a unified and rigorous $C^*$-algebraic framework for analyzing topological order and quantum memory in higher abelian quantum double models (HA-QDMs). We extend the constructions of \cite{dealmeida2017topological} to infinite $N$-dimensional simplicial complexes in a way that guarantees stability, in the spirit of the results of \cite{Vrana} in $N$ dimensions, while accommodating both the inter-level and intra-level gauge relations absent in previous works.  Here, by \emph{stability} we mean that the interaction defining the HA-QDM dynamics gives rise to a well-defined derivation on the local observables, which is closable and whose closure generates a strongly continuous one-parameter group of $*$-automorphisms of the full quasi-local $C^*$-algebra $\mathfrak{A}$ (see Section \ref{sec:dyn}); this is what allows the construction to be carried out directly on infinite simplicial complexes, rather than being restricted to the finite complexes underlying compact manifolds, as in \cite{Lorca2025,dealmeida2017topological}.
	
	\medskip

	Beyond the mere existence and classification of the frustration-free ground states, one of the central achievements of this work is to show that their finer structural features are entirely governed by the topology of the underlying (co)homological data. More precisely, when the relevant (co)homology groups $\mathring{H}^0(\f{C},\n{G})$,  $\mathring{H}_0(\f{C},\n{G})$ are finite, the logical algebra decomposes as $\mathfrak{A}_{\rm log}\simeq C(X_c)\otimes\mathcal{B}(\mathfrak{h}_q)$. In this case, the homological and cohomological groups, together with their canonical evaluation pairing, completely determine the resulting classical and quantum degrees of freedom. When both groups are trivial, the situation becomes even more rigid: not only is the frustration-free ground state unique, but, more deeply, the core algebra $\mathfrak{K}$ turns out to be a $C^*$-diagonal of $\mathfrak{A}$, so that $(\mathfrak{A},\mathfrak{K})$ forms a genuine Cartan pair. These facts can be summarized in the following theorem, which is established through the results proved in  Sections \ref{sec:analysis} and \ref{sec:diagonal}.
	
	\begin{teo}[Topological properties of the HA-QDMs]\label{teo:topological_char} 
		The (co)homological data encoded in $\mathring{H}^0(\f{C},\n{G})$ and $\mathring{H}_0(\f{C},\n{G})$ determine the degeneracy of the frustration-free ground state space of the HA-QDMs. Moreover, the triviality of these groups not only ensures the uniqueness of the frustration-free ground state, but, more deeply, completely characterizes the existence of a relevant $C^*$-diagonal inside the full algebra of observables, revealing a Cartan pair structure.
	\end{teo}
	
	This dichotomy is not merely structural. For planar Kitaev models, the existence of such a Cartan pair was recently exploited in \cite{PoloOjitoProdan2026} to obtain the complete classification of the KMS states of the model at every inverse temperature. We expect that an analogous strategy, relying on Theorem \ref{teo:topological_char} above, should extend this classification of KMS states to the higher-dimensional, higher-gauge setting studied here.
	
	\medskip
	
	The paper is organized as follows. Section \ref{sec:description} introduces the geometric and algebraic foundations of HA-QDMs, including the construction of the operator algebra of observables and the precise definition of the Hamiltonian. Section \ref{sec:analysis} characterizes the frustration-free ground states using $C^*$-algebra techniques, classifies them via homology and cohomology groups, and discusses the resulting implications for quantum computation and error correction. Section \ref{sec:diagonal} establishes the topological criterion, stated above, for $\mathfrak{K}$ to be a $C^*$-diagonal of $\mathfrak{A}$. Three appendices collect, respectively, the relation with Kitaev's original quantum double models (Appendix \ref{sec:quantum_double_models}), further examples recovered by our construction (Appendix \ref{sec:models}), and the background material on simplicial complexes used throughout the paper (Appendix \ref{sec:simplicial_complex}). Overall, the results establish a robust algebraic and topological foundation for understanding topological order and quantum memory in higher dimensions.

	\medskip
	
	\noindent
	{\bf Acknowledgments.} The research of J.A.F is supported by the \emph{European Union (ERC Starting Grant MathQuantProp, Grant Agreement 101163620)}\footnote{Views and opinions expressed are however those of the authors only and do not necessarily reflect those of the European Union or the European Research Council Executive Agency. Neither the European Union nor the granting authority can be held responsible for them.}. The author also gratefully acknowledges the support of the grant \emph{ANID Beca Magister Nacional folio 22221193}.
	GD's research is supported by the grant \emph{Fondecyt Regular - 1230032}. J.L.E. thanks the \emph{FONDECYT grant N°11241170} for the support during this research.

	\medskip
	\noindent
	{\bf Conflict of interest statement.} On behalf of all authors, the corresponding author states that there is no conflict of interest.

	\section{\label{sec:description}Description of the Model }

	\subsection{Geometric and algebraic data}\label{sec:basic_F}
	Let $\f{K}$ be a simplicial complex, where $\f{K}_n\subset \f{K}$ denotes the set of the $n$-simplices
	and $n\in \n{N}_0:=\n{N}\cup\{0\}$. 
	We will say that $\f{K}$ is a simplicial $N$-complex if $\f{K}_n=\emptyset$ for all $n>N$ and $N$ is referred to as the \emph{dimension} of the complex. We define the \emph{order} of a simplex $x \in \f{K}_n$ as the number of simplices $ y \in \f{K}_{n+1}$ that contain $x$ as a face. We say that the simplicial complex $\f{K}$ is \emph{$L$-regular} if, for all $x \in \f{K}$, the order of $x$ does not exceed  $L$. Throughout the paper, $\f{K}$ is assumed to be a countable, $N$-dimensional and $L$-regular
	simplicial complex with finite  $N,L\in\n{N}$. Appendix \ref{sec:simplicial_complex} contains a brief 
	summary of the main properties of the theory of simplicial complexes.

	\medskip

	Let $\f{C}_n$ be the \emph{free abelian group} with basis $\f{K}_n$ and $(\f{C}_{\bullet},\partial_\bullet^\f{C})$
	the associated chain complex. We will refer to $(\f{C}_{\bullet},\partial_\bullet^\f{C})$ as the \emph{geometric data}.
	We also define a second chain complex $(\n{G}_\bullet ,\partial_\bullet ^\n{G})$ where $\n{G}_\bullet:=\{\n{G}_j\}_{j\in\n{Z}}$
	is a sequence of \emph{finite abelian groups} linked by  groups morphisms $\partial_{j}^\n{G}: \n{G}_j\to \n{G}_{j-1}$ such that   $\partial_j^\n{G} \circ \partial_{j+1}^\n{G}=0$. 
	We also need the compatibility conditions $\n{G}_j=0$ if $j>N$ or if $j<0$.  
	We will call 
	$(\n{G}_\bullet ,\partial_\bullet ^\n{G})$ the \emph{algebraic data}.
	
	\medskip

	\begin{remark}[On the abelian assumption]\label{rk:abelian}
		The requirement that each $\n{G}_j$ be abelian is used essentially, and not merely for convenience. Localized maps of the form $x^*_g$ (Definition \ref{def:localized maps}) already form a group under addition only because the target $\n{G}_{n-p}$ is abelian; if the $\n{G}_j$ were allowed to be non-abelian, the set of $p$-maps ${\rm Hom}(\f{C},\n{G})^p$ introduced above would in general fail to be a group at all, and one would have to work with families of morphisms without a linear structure.\\
		The abelian assumption is also required for the Pontryagin-duality
		of the next section (more on this below).
	\end{remark}
	
	\medskip
	
	\begin{definition}[$p$-maps]
		Let $p\in\n{Z}$. A $p$-map $t:(\f{C}_{\bullet},\partial_\bullet^\f{C})\to(\n{G}_\bullet,\partial_\bullet^\n{G})$ is a sequence of morphisms $t_n:\f{C}_n\to \n{G}_{n-p}$. A $p$-map is also referred to as a \textit{configuration}. The set of all $p$-maps will be denoted by ${\rm Hom}(\f{C},\n{G})^p$.
	\end{definition}

	\medskip

	One has that  ${\rm Hom}(\f{C},\n{G})^p$  is actually an abelian group under the sum defined by $(f+t)_n=f_n+t_n$. The unit of this group is the trivial $p$-map, denoted by $0$, defined by the trivial morphisms $0_n:\f{C}_n\to \n{G}_{n-p}$. 
	Diagrammatically, an element $t\in{\rm Hom}(\f{C},\n{G})^0$ and an element $g\in{\rm Hom}(\f{C},\n{G})^1$ can be represented as
	\begin{figure}[H]
		\centering
		\begin{tikzcd}[row sep=1cm, column sep=2cm]
			\dots\arrow{r} & \f{C}_{n+1}\arrow{r}{\partial_{n+1}^\f{C}}\arrow{d}{t_{n+1}}\arrow{dr}{g_{n+1}}& \f{C}_n\arrow{r}{\partial_n^\f{C}}\arrow{d}{t_n}\arrow{dr}{g_n}&  \f{C}_{n-1}\arrow{r}{\partial_{n-1}^\f{C}}\arrow{d}{t_{n-1}}\arrow{dr}{g_{n-1}}&\dots\\
			\dots\arrow{r}& \n{G}_{n+1}\arrow{r}[swap]{\partial_{n+1}^\n{G}}& \n{G}_n\arrow{r}[swap]{\partial_n^\n{G}}&  \n{G}_{n-1}\arrow{r}[swap]{\partial_{n-1}^\n{G}}&\dots
		\end{tikzcd}
		\caption{A $0-$map $t$ (vertical arrows) and a $1-$map $g$ (diagonal arrows)}
		\label{fig:configurations}
	\end{figure}
	\begin{definition}[Cohomology]\label{def:delta^p map}
		Let  $\delta^{p}:{\rm Hom}(\f{C},\n{G})^{p}\to {\rm Hom}(\f{C},\n{G})^{p+1}$ be the map defined by
		\begin{equation*}
			(\delta^{p}t)_n\;:=\;t_{n-1}\circ\partial_n^\f{C}-(-1)^{p}\partial_{n-p}^\n{G}\circ t_n\;.
		\end{equation*}
		It results that $( {\rm Hom}(\f{C},\n{G})^\bullet  , \delta^\bullet )$ is a cochain complex, namely
		\begin{align}\label{pseudochain}
			\cdots \xleftarrow{} \;{\rm Hom}(\f{C},\n{G})^{p+2} \;\xleftarrow{\delta^{p+1}}\; {\rm Hom}(\f{C},\n{G})^{p+1} \;\xleftarrow{\delta^{p}}\; {\rm Hom}(\f{C},\n{G})^{p}\; \xleftarrow{} \cdots \quad ,
		\end{align}
		such that $\delta^{p+1} \circ \delta^p=0$. We will denote with 
		$H^p(\f{C},\n{G}):={\rm Ker}{(\delta^p)}/{\rm Im}{(\delta^{p-1})}$	
		the associated cohomology groups.
	\end{definition}

	\medskip

	We will also need to consider the subsets  ${\rm Hom}_0(\f{C},\n{G})^p$ consisting of all \emph{compactly supported} $p$-maps $t$ such that for any $n\in\n{N}_0$ it holds that $t_n(x)=0$ except for a finite number of $x\in \f{K}_n$. The set of $x\in \f{K}$ such that $t_n(x)\neq0$ for some $n\in\n{N}_0$
	will be called the \emph{support} of $t$.
	Evidently ${\rm Hom}_0(\f{C},\n{G})^p\subset {\rm Hom}(\f{C},\n{G})^p$ and one can restrict the map $\delta^p$ to this subspace. It results that $({\rm Hom}_0(\f{C},\n{G})^p,\delta^\bullet)$ is a cochain complex 
	\begin{align}\label{finite cochain}
		\cdots\; \xleftarrow{}\; {\rm Hom}_0(\f{C},\n{G})^{p+2} \;\xleftarrow{\delta^{p+1}}\; {\rm Hom}_0(\f{C},\n{G})^{p+1}\;\xleftarrow{\delta^{p}} \;{\rm Hom}_0(\f{C},\n{G})^p \;\xleftarrow{} \cdots \;\quad \;,
	\end{align}
	and the associated \emph{finite} cohomology groups are defined as usually as  $\mathring{H}^{p}(\f{C},\n{G}):={\rm Ker}( \delta^p)/{\rm Im}( \delta^{p-1})$.

	\medskip

	In the relevant case where $\f{K}$ is finite (this is the case for topological quantum computation \cite{Pachos}) one has the coincidence $\mathring{H}^{p}(\f{C},\n{G})=H^p(\f{C},\n{G})$.
	In this situation, as a consequence of an important result by Brown in \cite{brown}, one gets the isomorphisms 
	\begin{equation}\label{mainthr}
		H^p(\f{C},\n{G})\;\simeq\; \prod_{n\in\n{N}_0}H^n(\f{C}, H_{n-p}(\n{G}))\;,\qquad p\in\n{Z}\;,
	\end{equation}
	where $H_{\bullet}(\n{G})$ denotes the homology of $\n{G}$ and $\prod$ is the direct product.
	A detailed proof of 
	this result can be found in  \cite[Appendix B]{dealmeida2017topological}.

	\medskip

	The isomorphism \eqref{mainthr}  highlights a powerful way to compute the cohomology $H^\bullet(\f{C},\n{G})$ once one knows how to handle the groups $H^n(\f{C}, H_{j}(\n{G}))$. The latter can be computed by using the 
	Universal Coefficient Theorem \cite[Theorem 3.2]{Hatcher} which provides
	\begin{equation}\label{UCT}
		H^n(\f{C}, H_{j}(\n{G}))\;\simeq\;{\rm Hom}(H_n(\f{C}),H_j(\n{G}))\;\oplus\; {\rm Ext}^1(H_{n-1}(\f{C}),H_j(\n{G}))
	\end{equation}
	for every $n\in\n{N}_0$ and $j\in\n{Z}$. Here $H_\bullet(\f{C})$
	denotes the homology of the chain complex $\f{C}$ (of free abelian groups).

	\subsection{Characters and dualization}
	An important role will be played by the dual of the construction above. For this, first of all, we take the (Pontryagin) dual groups  $\widehat{\n{G}}_j$ consisting of the set of irreducible characters of $\n{G}_j$. These are abelian group under point-wise multiplication \cite[Exercise 2.7]{issacs}, but for convenience we will use the additive notation, namely $\alpha+\beta:=\alpha\cdot \beta$ and $-\alpha:=\alpha^{-1}$ for $\alpha,\beta\in \widehat{\n{G}}_j$.
	Let us recall that for any $g\in\n{G}_j$ and $\alpha\in \widehat{\n{G}}_j$ one gets that $\alpha(g)\in\n{U}(1)$ where we are identifying 
	$\n{U}(1)$ with the set of complex numbers with modulus 1.

	\medskip

	The adjoint of the  boundary map $\partial_{j}^\n{G}$ is the map  $\widehat{\partial}_{j}^\n{G}:\widehat{\n{G}}_{j-1}\to\widehat{\n{G}}_{j}$ defined by $\widehat{\partial}_{j}^\n{G}(\alpha):=\alpha\circ \partial_{j}^\n{G}$. We will denote with $(\widehat{\n{G}}_\bullet ,\widehat{\partial}_\bullet ^\n{G})$ the cochain complex associated by dualizing   the algebraic data
	$(\n{G}_\bullet ,\partial_\bullet ^\n{G})$. We need to introduce the dual version of the 
	$p$-maps defined above. For that we need the  coboundary operator $\widehat{\partial}^\f{C}_{n+1}:\f{C}_n\to \f{C}_{n+1}$ defined in Appendix \ref{sec:simplicial_complex}.

	\begin{definition}[Dual $p$-maps]
		Let $p\in\n{Z}$. A dual $p$-map  $\gamma:(\f{C}_{\bullet},\widehat{\partial}_\bullet^\f{C})\to(\widehat{\n{G}}_\bullet,\widehat{\partial}_\bullet^\n{G})$ is a sequence of morphisms $\gamma_n:\f{C}_n\to \widehat{\n{G}}_{n-p}$. A dual $p$-map is also referred to as a \textit{representation}. The set of all dual $p$-maps will be denoted by ${\rm Hom}(\f{C},\n{G})_p$.
	\end{definition}

	\begin{remark}[Notation]
		Let us point out that we will use Latin letters for elements of ${\rm Hom}(\f{C},\n{G})^p$ (configurations) and Greek letters for the elements of ${\rm Hom}(\f{C},\n{G})_p$ (representations). \hfill $\blacktriangleleft$
	\end{remark}

	Analogously to the $p$-maps, also the dual $p$-maps can be represented in a diagrammatic way. For example, a dual $1-$map can be represented as the following diagram:
	\begin{figure}[!h]
		\centering
		\begin{tikzcd}[scale=1,row sep=1cm, column sep=2cm]
			\dots\arrow{dr}{\gamma_{n+2}}&\arrow{l}  \f{C}_{n+1}\arrow{dr}{\gamma_{n+1}}&\arrow{l}[swap]{\widehat{\partial}_{n+1}^\f{C}} \f{C}_n\arrow{dr}{\gamma_n} &\arrow{l}[swap]{\widehat{\partial}_n^\f{C}}\f{C}_{n-1} \arrow{dr}{\gamma_{n-1}}&\arrow{l}[swap]{\widehat{\partial}_{n-1}^\f{C}}\dots\\
			\dots&\arrow{l} \widehat{\n{G}}_{n+1}&\arrow{l}{\widehat{\partial}_{n+1}^{\n{G}}} \widehat{\n{G}}_n&\arrow{l}{\widehat{\partial}_n^{\n{G}}} \widehat{\n{G}}_{n-1}&\arrow{l}{\widehat{\partial}_{n-1}^{\n{G}}} \dots
		\end{tikzcd}
		\caption{\label{fig:dual map} dual 1-map}
	\end{figure}

	\begin{definition}[Homology]\label{def:delta_p map}
		Let  $\delta_{p}:{\rm Hom}(\f{C},\n{G})_{p}\to {\rm Hom}(\f{C},\n{G})_{p-1}$ be the map defined by
		\begin{equation*}
			(\delta_p\gamma)_n\;:=\;(\gamma_{n+1}\circ \widehat{\partial}_{n+1}^\f{C}) -(-1)^{p-1} (\widehat{\partial}^\n{G}_{n+1-p}\circ \gamma_n)\;.
		\end{equation*}
		It results that $( {\rm Hom}(\f{C},\n{G})_\bullet  , \delta_\bullet )$ is a chain complex, namely
		\begin{align}\label{chain complex}
			\cdots \xrightarrow{} {\rm Hom}(\f{C},\n{G})_{p+1} \xrightarrow{\delta_{p+1}} {\rm Hom}(\f{C},\n{G})_{p} \xrightarrow{\delta_{p}} {\rm Hom}(\f{C},\n{G})_{p-1} \xrightarrow{} \cdots \quad ,
		\end{align}
		such that $\delta_{p} \circ \delta_{p+1}=0$. We will denote with 
		$H_p(\f{C},\n{G}):={\rm Ker}({\delta_p})/{\rm Im}({\delta_{p+1}})$	
		the associated homology groups.
	\end{definition}

	\medskip

	Let ${\rm Hom}_0(\f{C},\n{G})_{p}$ be the set of all \emph{compactly supported} dual $p$-maps $\gamma$ such that for any $n\in\n{N}_0$ it holds that $\gamma_n(x)=0$ (the trivial character) except for a finite number of $x\in \f{K}_n$. The set of $x$ 
	such that $\gamma_n(x)\neq 0$ for some $n\in\n{N}_0$ will be called the \emph{support} of $\gamma$.
	By definition ${\rm Hom}_0(\f{C},\n{G})_{p}$ is a subset of ${\rm Hom}(\f{C},\n{G})_{p}$ and one can restrict the dual map $\delta_p$ to this subspace. It results that 
	$( {\rm Hom}_0(\f{C},\n{G})_{\bullet}  , \delta_\bullet )$ is a chain complex:

	\begin{align}%\label{chain complex}
		\cdots \xrightarrow{} {\rm Hom}_0(\f{C},\n{G})_{p+1} \xrightarrow{\delta_{p+1}} {\rm Hom}_0(\f{C},\n{G})_{p} \xrightarrow{\delta_{p}} {\rm Hom}_0(\f{C},\n{G})_{p-1} \xrightarrow{} \cdots \quad ,
	\end{align}
	and the associated \emph{finite} homology groups are $\mathring{H}_{p}(\f{C},\n{G}):={\rm Ker}({\delta_p})/{\rm Im}({\delta_{p+1}})$ as usual.

	\medskip

	Let $t\in {\rm Hom}(\f{C},\n{G})^{p}$ and $\gamma\in {\rm Hom}(\f{C},\n{G})_{p}$ (with at least one of them of finite support). The representation of the configuration $t$ by $\gamma$
	is defined by
	\begin{equation}\label{eq:g_f}
		\gamma(t)\;:=\; \prod_{x\in \f{K}}[\gamma(t)]_x\;\in\;\n{U}(1)
	\end{equation}
	where for every $x\in \f{K}_n$
	\[
	[\gamma(t)]_x\;:=\;\gamma_n(x)[t_n(x)]\;\in\;\n{U}(1)
	\;,
	\]
	meaning that the character $\gamma_n(x)\in\widehat{\n{G}}_{n-p}$ is evaluated on the group element $t_n(x)\in {\n{G}}_{n-p}$.

	\medskip
	
	The proof of the following important result is postponed to Appendix \ref{sec:simplicial_complex} 
	\begin{lemma}\label{lem:duality}
		Let $t\in {\rm Hom}(\f{C},\n{G})^{p}$ and $\gamma\in {\rm Hom}(\f{C},\n{G})_{p+1}$, with at least one of them of finite support. Then, the duality
		\begin{equation*}
			\gamma(\delta^p t)\;=\;(\delta_{p+1}\gamma)(t)
		\end{equation*}
		holds true.
	\end{lemma}

	\subsection{Operator Algebra}\label{sec:Al-Op}
	In this section, we will define an algebra that contains the relevant observables for our study.
	For each $n\in\n{N}_0$ and  $x\in \f{K}_n$ consider the (finite dimensional) Hilbert space $\mathfrak{h}_x:=L^2(\n{G}_n)\simeq \n{C}^{|\n{G}_n|}$
	endowed with the canonical
	orthonormal basis  $\{\ket{g}\}_{g\in \n{G}_n}$. Let $\f{P}(\f{K})$ be the collection of finite subsets of $\f{K}$. Let us point out that every element of $\f{P}(\f{K})$ is a finite collection of vertices, edges, faces, and so on.
	For every $\Lambda\in \f{P}(\f{K})$ 
	define 
	\begin{equation*}
		\mathfrak{H}_\Lambda\;:=\;\bigotimes_{x\in \Lambda}\mathfrak{h}_x
	\end{equation*}
	and $\mathfrak{A}_\Lambda:=\mathcal{B}(\mathfrak{H}_\Lambda)$ the $C^*$-algebra of bounded operators on $\mathfrak{H}_\Lambda$. If $\Lambda_1\subset \Lambda_2$ there is a natural inclusion $i_{\Lambda_2,\Lambda_1}:\mathfrak{A}_{\Lambda_1}\hookrightarrow \mathfrak{A}_{\Lambda_2}$ mapping $A\mapsto A\otimes \n{1}_{\Lambda_2\backslash\Lambda_1} $. The maps $i_{\Lambda_2,\Lambda_1}$ are isometric morphisms so we will abuse notation and identify $i_{\Lambda_2,\Lambda_1}(A)$ with $A$. For this net of algebras, one defines the \emph{local} and  \emph{quasi-local} algebras of observables as
	\begin{equation*}
		\mathfrak{A}_{\rm loc}\;:=\;\bigcup_{\Lambda\in \f{P}(\f{K})}\mathfrak{A}_\Lambda\;,\qquad\quad\mathfrak{A}\;:=\;\overline{\mathfrak{A}_{\rm loc}}^{\;\norm{\cdot}}\;,
	\end{equation*}
	respectively. 
	An observable $A\in \mathfrak{A}$ is said to be supported on $\Lambda$ if $A\in \mathfrak{A}_\Lambda$. The smallest set $\Lambda$ such that $A\in \mathfrak{A}_\Lambda$ is called the support of $A$. For $A\in \mathcal{B}(\mathfrak{h}_x)$ we use the notation $A_x$ to refer to the corresponding operator acting at the site $x$.
	It can be shown that $\mathfrak{A}$ is in fact a UHF algebra \cite[Example 2.6.12]{Bratelli}.

	\subsection{Local operators}
	The models of interest in this work will be defined through an important pair of building block operators.
	For every  $x\in \f{K}_n\subset \f{K}$ consider $\mathfrak{h}_x$ endowed with its canonical
	orthonormal basis  $\{\ket{h}\}_{h\in \n{G}_n}$.
	For $g\in \n{G}_n$ and $\alpha\in \widehat{\n{G}}_n$ consider the operators $P_g$ and $Q_\alpha$
	defined on the orthonormal basis of $\mathfrak{h}_x$
	by
	\begin{equation*}
		P_g\ket{h}\;:=\:\ket{g+h}\;,\qquad\quad Q_{\alpha}\ket{h}\;:=\;\alpha(h)\ket{h}
	\end{equation*}

	\begin{remark}[Weyl matrices]\label{rk_lin_bas_mat}
		These unitary operators can be seen as a generalization of the usual Pauli matrices $\sigma^1$ and $\sigma^3$ acting on $\n{C}^2$.
		More precisely one can check that elements of the form $P_gQ_{\alpha}$ with $g\in \n{G}_n$ and $\alpha\in \widehat{\n{G}}_n$, usually called  \emph{Weyl matrices}, provide a linear basis for the full matrix algebra
		$\f{B}(\mathfrak{h}_x)\simeq{\rm Mat}_{|{\n{G}}_n|}(\n{C})$. See \cite[Lemma 3.2]{weil_operators}.
		\hfill $\blacktriangleleft$
	\end{remark}

	\medskip

	In order to generalize and extend the definition of the building block operators above, we need to introduce
	certain specific maps that only act non-trivially on one element of the simplicial complex.
	\begin{definition}[Localized Maps]\label{def:localized maps}
		For every $x\in \f{K}_n$, $g\in \n{G}_{n-p}$ define $x^*_g\in {\rm Hom}_0(\f{C},\n{G})^p$ as 
		\begin{equation*}
			x^*_g(y)\;:=\;\left\{ \begin{array}{lcc} g & {\rm if} & x=y\;, \\ 
				0 & {\rm if} & x\not=y\;. \\  \end{array} \right.
		\end{equation*}
		Analogously, for every $\alpha\in \widehat{\n{G}}_{n-p}$ define $x_*^\alpha \in {\rm Hom}_0(\f{C},\n{G})_{-p}$ as
		\begin{equation*}
			( x_*^\alpha)(y)\;:=\;\left\{ \begin{array}{lcc} \alpha & {\rm if} & x=y\;, \\  0 & {\rm if} & x\not=y\;. \\   \end{array} \right.
		\end{equation*}
	\end{definition}

	\medskip

	As the name indicates, the maps defined above are localized around $x\in \f{K}_n$, and can be used to define local operators.
	First of all, notice that any element $t\in {\rm Hom}_0(\f{C},\n{G})^0$ can be written as a finite sum of $0$-maps of the form $x^*_{g_x}$ where for every $x\in \f{K}_n$ the 
	element  $g_x\in \n{G}_n$ is defined by $g_x:=t(x)$. Then, it turns out that
	\begin{equation*}
		t\;=\;\sum_{x\in \f{K}} x^*_{g_x}
	\end{equation*}
	and the associated \textit{shift  operator} is given by
	$$P_t\;:=\;\prod_{x\in \f{K}}P_{ x^*_{g_x}}\;.
	$$
	Analogously, every $\gamma\in {\rm Hom}_0(\f{C},\n{G})_0$ can be  written as 
	\begin{equation*}
		\gamma\;=\;\sum_{x\in \f{K}} x_*^{\alpha_x}
	\end{equation*}
	with $\alpha_x:=\gamma(x)$. The associated \textit{clock  operator} is then given by
	$$
	Q_\gamma\;:=\;\prod_{x\in \f{K}} Q_{x_*^{\alpha_x}}\;.
	$$
	In particular, one has that $P_0=Q_0=\n{1}$ where   $0$ denotes the trivial homomorphisms. 
	By a direct computation (see \cite[Appendix D]{dealmeida2017topological}) one can check the group relations
	\begin{equation}\label{relatP-Q}
		P_{t_1}P_{t_2}\;=\;P_{t_1+t_2}\;,\qquad\quad Q_{\gamma_1}Q_{\gamma_2}\;=\;Q_{\gamma_1+\gamma_2}
	\end{equation}
	for every $t_1,t_2\in {\rm Hom}_0(\f{C},\n{G})^0$ and every $\gamma_1,\gamma_2\in {\rm Hom}_0(\f{C},\n{G})_0$. Moreover, the product between a $P_t$ and a $Q_\gamma$ is fixed by the  Weyl-type 
	relation
	\begin{equation}\label{eq:commutation_P_and_Q}
		Q_\gamma P_t\;=\;\gamma(t)P_tQ_\gamma
	\end{equation}
	valid for every $t\in {\rm Hom}_0(\f{C},\n{G})^0$ and $\gamma\in {\rm Hom}_0(\f{C},\n{G})_0$
	with $\gamma(t)$ given by \eqref{eq:g_f}.
	In particular, the $P_t$ and $Q_\gamma$ are unitary operators.
	The following result will be used several times in this work.
	\begin{lemma}[Linear independent generators]\label{lemm:lin_indip_g}
		Every $R\in \mathfrak{A}_{\rm loc}$ can be expressed in a unique way as finite sum
		\[
		R\;=\;\sum_{i=1}^n\lambda_{i}P_{t_i}Q_{\gamma_i}
		\]
		for some $n\in\n{N}$ with $t_i\in {\rm Hom}_0(\f{C},\n{G})^0$, $\gamma_i\in {\rm Hom}_0(\f{C},\n{G})_0$ and $\lambda_i\in\n{C}\setminus\{0\}$ for every $i=1,\ldots,n$.
	\end{lemma}
	\begin{proof}
		Let $\Lambda$ be the support of $R$. Recall that if $x \in \f{K}_n$, the operators of the form $P_g Q_\gamma$ with $g \in \n{G}_n$ and $\alpha \in \widehat{\n{G}}_n$ form a basis for the full matrix algebra
		$\f{B}(\mathfrak{h}_x)\simeq{\rm Mat}_{|{\n{G}}_n|}(\n{C})$ (see Remark \ref{rk_lin_bas_mat}). 
		Since the tensor product of these basis elements forms a basis for the tensor product space, the operators $P_t Q_\gamma$ with support contained in \( \Lambda \) form a basis for \( \mathfrak{A}_\Lambda \). In particular, we can express $R$ in a unique way as a (finite) linear combination of operators in $\mathfrak{A}_\Lambda$.
		Now, suppose we can express $R$ as a finite sum of elements in $\mathfrak{A}_{\rm loc}$ in two different ways:
		\[
		R\;=\;\sum_{i=1}^n\lambda_{i}P_{t_i}Q_{\gamma_i}\;=\; \sum_{i=1}^m\lambda_{i}'P_{t_i}'Q_{\gamma_i}'
		\]
		Taking $\Lambda'$ big enough to contain the support of all the $P_{t_i}Q_{\gamma_i}$, all the $P_{t_i}'Q_{\gamma_i}'$, and considering the embedding of all this operators inside $\Lambda'$, one has that $R$ can be written in $\mathfrak{A}_{\Lambda'}$  as a linear combination of elements of the basis of $\mathfrak{A}_{\Lambda'}$ in two different ways. However, this is a contradiction. Therefore, for every $R\in \mathfrak{A}_{\rm loc}$, there is a unique expansion in the basis elements  $P_tQ_\gamma$.
	\end{proof}

	\medskip

	In view of the result above, one can say that $P_{t}Q_{\gamma}$ linearly spans $\mathfrak{A}_{\rm loc}$ and has norm-dense span in $\mathfrak{A}$.

	\medskip

	Using the operators $\delta^p$ and $\delta_p$ described in Definitions \ref{def:delta^p map} and \ref{def:delta_p map}, respectively, one can introduce two new families of unitary operators which will be the relevant elements for the construction of the Hamiltonian of the main model.
	For every $t\in {\rm Hom}_0(\f{C},\n{G})^{-1}$ define
	\begin{equation*}
		A_t\;:=\;P_{\delta^{-1}t}\;,
	\end{equation*}
	and for every $\gamma\in {\rm Hom}_0(\f{C},\n{G})_1$ define 
	\begin{equation*}
		B_\gamma:=Q_{\delta_1 \gamma}\;.
	\end{equation*}
	The new pairs of operators  $A_t,B_\gamma$ have the advantage of ``correcting'' the commutation rule \eqref{eq:commutation_P_and_Q} between $P_t$ and $Q_\gamma$.
	\begin{lemma}\label{lem:commut_AB}
		It holds true that
		\begin{equation*}
			[A_t,B_\gamma]\;=\;0
		\end{equation*}
		for every $t\in {\rm Hom}_0(\f{C},\n{G})^{-1}$ and every $\gamma\in {\rm Hom}_0(\f{C},\n{G})_1$.
	\end{lemma}

	\begin{proof}
		Using Eq.\eqref{eq:commutation_P_and_Q} and Lemma \ref{lem:duality} one has
		\begin{align*}
			B_\gamma A_t&=Q_{\delta_1\gamma}P_{\delta^{-1}t}\\
			&=\delta_1\gamma(\delta^{-1}t)P_{\delta^{-1}t}Q_{\delta_1\gamma}\\
			&=\gamma((\delta^0\circ \delta^{-1}) t) A_t B_\gamma
		\end{align*}
		Since $\delta^0\circ\delta^{-1}=0$, as expressed by  the cochain complex diagram \eqref{finite cochain}, one obtains $B_\gamma A_t=A_tB_\gamma$.
	\end{proof}

	\begin{definition}[Generalized star and plaquette operators]
		\label{def:local_operators}	Let $x\in \f{K}_n$. Define
		\begin{equation*}
			A_x^0\;:=\;\frac{1}{\ab{\n{G}_{n+1}}} \sum_{g\in \n{G}_{n+1}} A_{x^*_g}
		\end{equation*}
		\begin{equation*}
			B_x^0\;:=\;\frac{1}{\ab{\n{G}_{n-1}}}\sum_{\gamma\in \widehat{\n{G}}_{n-1}} B_{ x_*^\gamma}
		\end{equation*}
		where $|\n{G}_{n}|$ is the cardinality of $\n{G}_{n}$.
	\end{definition}

	\begin{lemma}\label{lemma:projections}
		For every $x\in \f{K}$, the operators $A_x^0$ and $B_x^0$ form a set of commuting orthogonal projections.
	\end{lemma}

	\begin{proof}
		The commutation follows from Lemma \ref{lem:commut_AB} along with the group properties 
		$A_{t_1}A_{t_2}=A_{t_1+t_2}$ and $B_{\gamma_1}B_{\gamma_2}=B_{\gamma_1+\gamma_2}$ that can be deduced directly by \eqref{relatP-Q}.
		A direct computation shows that
		\begin{align*}
			A_x^0A_x^0&=\frac{1}{\ab{\n{G}_{n+1}}^2}\sum_{g\in \n{G}_{n+1}}\sum_{h\in \n{G}_{n+1}}A_{x^*_g}A_{x^*_h}\\
			&=\frac{1}{\ab{\n{G}_{n+1}}^2}\sum_{g\in \n{G}_{n+1}}\sum_{h\in \n{G}_{n+1}}A_{x^*_{g+h}}\\
			&=\frac{1}{\ab{\n{G}_{n+1}}}\sum_{g\in\n{G}_{n+1}}A_x^0\\
			&=A_x^0
		\end{align*}
		Analogously, $\pa{B_x^0}^2=B_x^0$. It remains to show that these operators
		are self-adjoint. First, notice that $$(A_{x^*_g})^*\;=\;(P_{\delta^{-1}x^*_g})^*\;=\;P_{\delta^{-1}(-x^*_g)}\;.$$
		By the Definition \ref{def:localized maps} one has that $-x^*_g=x^*_{-g}$, and in turn
		$$(A_{x^*_g})^*\;=\;A_{x^*_{-g}}\;.$$ Since one is summing over all the elements of the group, one gets
		\begin{align*} 
			\pa{A_x^0}^*\;&=\;\frac{1}{\ab{\n{G}_{n+1}}}\sum_{g\in \n{G}_{n+1}}\pa{A_{x^*_{g}}}^*\\
			&=\;\frac{1}{\ab{\n{G}_{n+1}}}\sum_{g\in \n{G}_{n+1}}A_{x^*_{-g}}\\
			&=\;A_x^0\;.
		\end{align*}
		The same reasoning also shows that $\pa{B_x^0}^*=B_x^0$.
	\end{proof}

	\medskip

	The local operators defined above are the higher dimensional generalization of the \emph{star} and \emph{plaquette} operators defined in \cite{Kitaev_2003}. Observe, using Figures \ref{fig:configurations}, \ref{fig:dual map} and Definitions \ref{def:delta^p map},\ref{def:delta_p map}, that the \emph{supports} of \( A_x^0 \) and \( B_x^0 \) are given by

	\[
	\begin{aligned}
		{\rm supp}(A_x^0)\:&:=\; \{y\in\f{K}\;|\;(\delta^{-1}x^*_g)(y)\neq 0\}\;=\;\{x\}\cup \{z\in\f{K}\;|\; \epsilon_{z,x}\neq 0\}\\
		{\rm supp}(B_x^0)\:&:=\; \{y\in\f{K}\;|\;(\delta_{1}x_*^\gamma)(y)\neq 0\}\;=\;\{x\}\cup \{y\in\f{K}\;|\; \epsilon_{x,y}\neq 0\}
	\end{aligned}
	\] 
	where the numbers $\epsilon_{z,x}$ are defined in Appendix \ref{sec:simplicial_complex}.
	In particular, the support of $A_x^0$ is given by $x$ and its \emph{co-boundary}, while the support of $B_x^0$ is given by $x$ and its \emph{boundary}.

	\subsection{The HA-QDM Dynamics}\label{sec:dyn}
	As we are working with a spin-like system, we need to specify how local operators interact with each other.
	An interaction is a map $\Phi$ which associates to each $X\in\f{P}(\f{K})$ a  self-adjoint operator $\Phi_{\rm QDM}(X)\in\mathfrak{A}$. The relevant interaction for our model is defined as 
	\begin{equation*}
		\Phi_{\rm QDM}(X)\;:=\;\left\{ \begin{array}{clr}
			-A_x^0&\; \text{if} & X={\rm supp}(A_x^0)\;\;%supp A_x^0 
			\text{for some $x\in \f{K}$}\;, \\
			-B_x^0&\; \text{if} & X={\rm {\rm supp}(B_x^0)}\;\; \text{for some $x\in \f{K}$}\;,\\
			0&\;& \text{otherwise}\;,
		\end{array} \right.
	\end{equation*}
	and is called \emph{higher Abelian quantum double model} (HA-QDM) \emph{interaction}.
	The total energy due to all the interactions in a finite set $\Lambda$ is given by the local \emph{HA-QDM Hamiltonian} 
	\begin{equation}\label{eq:ham}
		H^\Lambda_{\rm QDM}\;:=\;\sum_{X\subset \Lambda} \Phi_{\rm QDM}(X)\;=\;-\left(\sum_{\substack{x\in \f{K}\\
				{\rm supp}(A_x^0)\subseteq \Lambda }}A_x^0+\sum_{\substack{y\in \f{K}\\
				{\rm supp}(B_y^0)\subseteq \Lambda }}B_y^0\right)\;.
	\end{equation}
	With this, one can define a derivation $\delta_{\rm QDM}$ with dense domain
	$\f{D}(\delta_{\rm QDM})=\mathfrak{A}_{\rm loc}$.
	For a local operator $A\in \mathfrak{A}_{\rm loc}$ one has that
	\begin{equation}\label{eq:deriv}
		\delta_{\rm QDM}(A)\;:=\;{\rm i}\lim_{\Lambda_n\uparrow \f{K}}\big[H_{\rm QDM}^{\Lambda_n},A\big]\;.
	\end{equation}
	Under some specific conditions on the simplicial complex, this derivation turns out to be closable, and its closure generates a dynamics. 
	For instance if it happens that
	\begin{equation}\label{ec:condiciones}
		\sup_{x\in \f{K}} \ab{ {\rm supp}(A_x^0)}\;<\;\infty\;,\qquad \text{and}\qquad \sup_{x\in \f{K}} \ab{{\rm supp}(B_x^0)}\;<\;\infty\;,
	\end{equation}
	\emph{i.e.} when there is a uniform bound on the size of the supports, 
	then the \cite[Theorem 6.2.4]{bratelli2} guarantees that $\delta_{\rm QDM}$ is closable and its closure (still denoted with the same symbol) generates a strongly continuous one-parameter group of $*$-automorphisms $\{\tau_t\}_{t\in\n{R}}$ of $\mathfrak{A}$. 
	Conditions \eqref{ec:condiciones} are automatically satisfied in the case of 
	$N$-dimensional $L$-regular
	simplicial complexes with finite  $N,L\in\n{N}$ (see Section \ref{sec:basic_F}).
	The dynamics generated by the derivation \eqref{eq:deriv} will be called the \emph{HA-QDM dynamics}.

	\subsection{Frustration-free ground state}\label{sec:frustration-free ground state}

	Once the dynamics have been defined, the next relevant notion is that of \emph{ground state}. A state is a positive linear functional $\omega:\mathfrak{A} \to \mathbb{C}$ normalized according to  $\omega(\n{1}) = 1$. The state space of $\mathfrak{A}$ is usually denoted with ${\mathtt{S}}(\mathfrak{A})$. It is a convex space
	whose extreme points are called \emph{pure states}. The set of pure states is denoted with $\mathtt{P}(\mathfrak{A})$. Both $\mathtt{S}(\mathfrak{A})$ and $\mathtt{P}(\mathfrak{A})$ are usually  topologized with the $\ast$-weak topology \cite[Section 2.3.2]{Bratelli}.

	\medskip

	A \emph{ground state} $\omega$ (for the dynamics defined by $\delta$) is a state that satisfies the following condition 
	\begin{equation*}
		-{\rm i}\;\omega(A^*\delta(A))\; \geqslant \;0\;,\qquad \forall \; A\in\mathfrak{A}_{\rm loc}\;.
	\end{equation*}

	\medskip

	From now on we will use the symbol
	$\mathtt{G}_{\rm QDM}(\mathfrak{A})$ to denote the space of ground states for  
	the {HA-QDM dynamics} $\delta_{\rm QDM}$ described in Section \ref{sec:dyn}. The subset of pure ground states will be denoted by $\mathtt{G}_{\rm QDM}^{\circ}(\mathfrak{A}):=\mathtt{G}_{\rm QDM}(\mathfrak{A})\cap\mathtt{P}(\mathfrak{A})$.

	\medskip

	An important role is played by states that are stabilized by all the fundamental operators $A_x^0$   and $B_x^0$.

	\begin{definition}[Frustration-free ground state]\label{def:ffgs}
		We say that a ground state $\omega_0$ (for a given $HA-QDM$ dynamics) is frustration-free if $\omega_0(A_x^0)=\omega_0(B_y^0)=1$ for all $x,y\in\f{K}$.
	\end{definition}

	\begin{remark}[GNS representation]
		Let $(\pi_0,\mathcal{H}_0,\psi_0)$ be the corresponding GNS-representation associated to a frustration-free ground state
		$\omega_0$.
		Then 
		$$\inner{\psi_0}{\pi_0(A_x^0)\psi_0}\;=\;1\;=\;\inner{\psi_0}{\pi_0(B_x^0)\psi_0}\;.$$ 
		By the Lemma \ref{lemma:projections}, $\pi_0(A_x^0)$ and $\pi_0(B_x^0)$ are projectors. Therefore, 
		\[
		\|\psi_0-\pi_0(A_x^0)\psi_0\|^2\;=\;\inner{\psi_0}{(\n{1}-\pi_0(A_x^0))\psi_0}\;=\;0
		\]
		showing that $\pi_0(A_x^0)\psi_0=\psi_0$. A similar argument also shows that $\pi_0(B_x^0)\psi_0=\psi_0$. In other words, $\psi_0$ is a common eigenvector (with eigenvalue 1) for all $\pi_0(A_x^0)$ and  $\pi_0(B_x^0)$.
		\hfill $\blacktriangleleft$
	\end{remark}

	\medskip

	The set of frustration-free elements of $\mathtt{G}_{\rm QDM}(\mathfrak{A})$ will be denoted with $\mathtt{FG}_{\rm QDM}(\mathfrak{A})$ and the related pure states with $\mathtt{FG}_{\rm QDM}^{{\circ}}(\mathfrak{A}):=\mathtt{FG}_{\rm QDM}(\mathfrak{A})\cap\mathtt{P}(\mathfrak{A})$.
	In Section \ref{sec:analysis} we will prove that $\mathtt{FG}_{\rm QDM}(\mathfrak{A})\neq \emptyset$, and provide a characterization of them.

	\section{Analysis of the space of Frustration-free ground states}\label{sec:analysis}

	With all the machinery developed previously, we are ready to determine and characterize the frustration-free ground states of the model introduced in Section \ref{sec:frustration-free ground state}. First of all, we will prove in Theorem \ref{theo:one_on_one_bijection} that these states are in one-to-one correspondence with states of a $C^*$-algebra $\mathfrak{A}_{\rm log}$. In this algebra, the states are completely determined by cohomology and homology groups. In this section, we will follow the work \cite{Vrana} since the results it contains are still valid in our context. However, unlike their setting, we cannot decompose the Hamiltonian described in Section \ref{sec:dyn} into a composition of Hamiltonians supported on disjoint dimensions. This is because our local operators $A_x^0$ and $B_x^0$ act across more than one dimension; see, for example, Figure \ref{fig:A_v^0 operator}. Consequently, we must work with the entire Hamiltonian at once. Nevertheless, in a certain sense, the ``dimension multiplication'' is encoded in the homology and cohomology
	through the isomorphism \eqref{mainthr}. In fact, when $\partial^\n{G}_n \equiv 0$, one recovers the result of  \cite{Vrana}.

	\medskip

	The fundamental principle underlying our approach is closely related to that used in stabilizer-based quantum error correction techniques \cite{Fujii2015}. In these methods, a designated set of \emph{stabilizers}, typically commuting observables that define the code space, is first identified. In our case, the stabilizers will be the local operators $A_x^0$ and $B_x^0$.
	Then, one considers the algebra of operators that commute with all elements of this stabilizer set, as these represent logical operations that preserve the code subspace. To capture the operational indistinguishability of such operators on a particular quantum state, we introduce an equivalence relation: two operators are considered equivalent if they produce the same expectation value when evaluated on a relevant state (frustration-free ground state). This construction not only reflects the symmetries of the system but also provides a natural framework for defining observables in scenarios where redundancy or gauge freedom is present.

	\subsection{Local characterization of frustration-free ground states}
	In this section, we will introduce a ``small'' abelian $C^*$-subalgebra
	$\mathfrak{K}\subset \mathfrak{A}$ endowed with a relevant pure state $\vartheta\in \mathtt{P}(\mathfrak{K})$ which has the crucial property of determining the elements of $\mathtt{FG}_{\rm QDM}(\mathfrak{A})$ by extensions. For these reasons we will refer to $\mathfrak{K}$ as the \emph{core} $C^*$-algebra and to $\vartheta$ as the \emph{seed} state of the
	{HA-QDM dynamics}.

	\medskip

	To start, let us recall a crucial result that has become foundational in the algebraic formulation of quantum double models. See \cite[Section 2.1.1]{Fannes} and \cite[Lemma 6.1]{Vrana} for the details of the proof.

	\begin{lemma}\label{lema estado}
		Let $\omega$ be a state in a $C^*$-algebra $\mathfrak{A}$. Let $U\in\mathfrak{A}$ be a unitary element such that $|\omega(U)|=1$. Then,  one has
		\begin{equation*}
			\omega(U)\omega(Y)\;=\;\omega(UY)\;=\;\omega(YU)
		\end{equation*}
		for any $Y\in\mathfrak{A}$.
	\end{lemma}

	\medskip

	This result is very useful  especially when $\omega(U)=1$ because it  allows us to ``introduce" $U$ inside the state, \emph{i.e.} $\omega(Y)=\omega(UY)=\omega(YU)$ for all $Y\in\mathfrak{A}$.

	\begin{prop}\label{prop:frist_cond}
		Let $\omega_0$ be a frustration-free ground state according to Definition \ref{def:ffgs}. Then $\omega_0(A_tB_\gamma)=1$ for all $t\in {\rm Hom}_0(\f{C},\n{G})^{-1}$ and $\gamma\in {\rm Hom}_0(\f{C},\n{G})_1$.
	\end{prop}
	\begin{proof}
		By definition $\omega_0(A_x^0)=\omega_0(B_y^0)=1$ holds. Combining this with the inequalities  $|\omega_0(A_{x_g^*})|\leqslant1$ and $|\omega_0(B_{x_*^\gamma})|\leqslant1$ due to the fact that the $A_{x_g^*}$ and $B_{x_*^\gamma}$
		are unitary,
		one gets that $\omega_0(A_{x_g^*})=1=\omega_0(B_{x_*^\gamma})$. Using that any operator $A_t$ can be written as a multiplication of operators of the form $A_{x_g^*}$ and any operator $B_\gamma$  can be written as a multiplication of operators of the form $B_{x_*^\gamma}$, and combining with  Lemma \ref{lema estado}, one concludes the proof.
	\end{proof}

	\medskip

	The next step is to prove the existence of frustration-free ground states. In particular, we will
	show that the condition  $\omega_0(A_tB_\gamma)=1$
	is sufficient to guarantee that  $\mathtt{FG}_{\rm QDM}(\mathfrak{A})\neq\emptyset$. To see this, first we need to introduce the following $C^*$-subalgebra.

	\begin{definition}[Core $C^*$-algebra]\label{def:K}
		Let $\mathfrak{K}:=C^*(A_t,B_\gamma)$ be the abelian $C^*$-subalgebra of $\mathfrak{A}$ generated by the operators $A_t$ and $B_\gamma$ for every $t\in {\rm Hom}_0(\f{C},\n{G})^{-1}$ and $\gamma\in {\rm Hom}_0(\f{C},\n{G})_1$.
	\end{definition}

	\medskip

	The $C^*$-algebra $\mathfrak{K}$ contains by construction the generalized star and plaquette operators $A_x^0$ and $B_x^0$.
	Let $\mathfrak{K}_0$ be the dense subalgebra given by the finite linear combinations of the monomials $A_tB_\gamma$. One has that 
	$\mathfrak{K}_0\subset \mathfrak{A}_{\rm loc}$ by construction.

	\medskip

	\begin{remark}\label{rk:generators_K}
		As this $\mathfrak{K}$ is abelian, product of element of the form $A_tB_\gamma$ is still of that form. Therefore $\mathfrak{K}_0$ is equal to ${\rm span}\set{P_{l}Q_{\zeta}\lvert (l,\zeta)\in {\rm Im}(\delta^{-1})\times {\rm Im}(\delta_1)}$. \hfill $\blacktriangleleft$
	\end{remark}

	\medskip

	On  $\mathfrak{K}$ let us introduce a relevant linear functional $\vartheta:\mathfrak{K}\to\n{C}$ by imposing the conditions
	\begin{equation}\label{eq:def_eta}
		\vartheta(A_tB_\gamma)\;=\;1\;, \qquad \forall \;(t,\gamma)\in {\rm Hom}_0(\f{C},\n{G})^{-1}\times {\rm Hom}_0(\f{C},\n{G})_1
	\end{equation}
	and extending  $\vartheta$ by linearity on the dense subalgebra $\mathfrak{K}_0$, and then by continuity on the 
	full algebra 
	$\mathfrak{K}$.

	\begin{lemma}[Seed state]\label{lemm:pur_st_K}
		Conditions \eqref{eq:def_eta} define a unique  pure state
		$\vartheta:\mathfrak{K}\to\CC$.
	\end{lemma}
	\begin{proof}

		Consider an increasing sequence $\Lambda_1\subset \Lambda_2\dots$ of finite subsets of $\f{K}$ such that their union is equal to $\f{K}$. For $\Lambda_n$ define the $C^*$-algebra $\mathfrak{K}_{n}$ as the $C^*$-algebra generated by the operators $A_t$ and $B_\gamma$ that have their support inside $\Lambda_n$ and consider $\vartheta_n:=\vartheta|_{\mathfrak{K}_n}$ the linear functional generated by the conditions 
		\eqref{eq:def_eta} restricted to $\mathfrak{K}_n$.
		Since $\mathfrak{K}_n$ is a finite-dimensional algebra generated by the elements  $A_tB_\gamma$ with support inside $\Lambda_n$, it follows that $\vartheta_n$ defines a state, and in particular it is continuous. Indeed, by the group relations \eqref{relatP-Q} one has $(A_tB_\gamma)(A_{t'}B_{\gamma'})=A_{t+t'}B_{\gamma+\gamma'}$, so the assignment $A_tB_\gamma\mapsto 1$ for every $(t,\gamma)$ is a group homomorphism from the (abelian) group of these unitaries into $\n{U}(1)$; extended by linearity to $\mathfrak{K}_n$. Therefore, it is a character of the finite-dimensional abelian $C^*$-algebra $\mathfrak{K}_n$, and every character of a $C^*$-algebra is automatically a (pure) state, i.e.\ positive and normalized.
		By the Hahn-Banach theorem \cite[Proposition 2.3.24]{Bratelli}, we can extend this state to a state in the full $C^*$-algebra $\mathfrak{K}$, and we will denote this state by $\vartheta_n^\sim$. This gives us a sequence of states $\{\vartheta_n^\sim\}_{n\in\NN}$ and since the state space of $\mathfrak{K}$ is $*$-weak compact, there is a subsequence $\{\vartheta_{n_l}^\sim\}_{l\in\NN}$ converging in the $\ast$-weak topology to a state on $\mathfrak{K}$. For $(t,\gamma)\in{\rm Hom}_0(\f{C},\n{G})^{-1}\times {\rm Hom}_0(\f{C},\n{G})_1$ this state maps $A_tB_\gamma$ to 1, because the states $\vartheta_n$ satisfy $\vartheta_n(A_tB_\gamma)=1$ for $n$ big enough. Therefore, we will denote this state by $\vartheta$. 
		We have proved that there exists (at least) one state on $\mathfrak{K}$ 
		that meets the conditions \eqref{eq:def_eta}. Now, let us prove the uniqueness. To do this, suppose that there exists another state $\vartheta'$ with this property. Consider $F\in\mathfrak{K}_0$ and write it as
		$$F\;=\;\sum_{i=1}^n \lambda_i A_{t_i}B_{\gamma_i}$$
		then,
		\begin{equation*}
			\vartheta(F)=\sum_{i=1}^n\lambda_i=\vartheta'(F)
		\end{equation*}
		Here, we used Lemma \ref{lemm:lin_indip_g} to guarantee that $F$ can be written in a unique way as a sum, since $A_{t_i}B_{\gamma_i}=P_{\delta^{-1}t_i}Q_{\delta_1\gamma_i}$ are linearly independent generators.
		So, $\vartheta$ and $\vartheta'$ coincide in the dense set $\mathfrak{K}_0$. Since there are states, in particular, they are continuous. Therefore, $\vartheta$ and $\vartheta'$ are the same state. 
		Let us show now that $\vartheta:\mathfrak{K}_0\to\CC$ is multiplicative. To see this, take $F,G\in \mathfrak{K}_0$ given by
		$$F\;=\;\sum_{i=1}^n \lambda_i A_{t_i}B_{\gamma_i}\;,\qquad G\;=\;\sum_{j=1}^m \lambda'_j A_{t_j'}B_{{\gamma_j'}}\;$$
		with $\lambda_i,\lambda'_j\in\n{C}$.
		Then,
		\begin{align*}
			\vartheta(FG)\;&=\;\vartheta\pa{\pa{\sum_{i=1}^n \lambda_i A_{t_i}B_{\gamma_i}}\pa{\sum_{j=1}^m \lambda'_j A_{{t_j'}}B_{{\gamma'_j}}}}\\
			&=\;\vartheta\pa{\sum_{i=1}^n\sum_{j=1}^m \lambda_i\lambda'_j A_{t_i+{t_j'}}B_{\gamma_i+{\gamma_j'}}}\;=\;
			\vartheta\pa{\sum_{i=1}^n\sum_{j=1}^m \lambda_i\lambda'_j }\\
			&=\;\pa{\sum_{i=1}^n \lambda_i}\pa{\sum_{j=1}^m \lambda_j'}\;=\;\vartheta(F)\vartheta(G)\;.
		\end{align*}
		where, in the second equality, one uses crucially the commutativity of the $A_t$'s and $B_\gamma$'s.
		Now, for arbitrary $F,G\in\mathfrak{K}$, consider $(F_n)_{n\in\NN}$ and $(G_n)_{n\in\NN}$ with $F_n,G_n\in\mathfrak{K}_0$. Then,
		$$\vartheta(FG)\;=\;\lim_{n\to\infty}\lim_{m\to\infty} \vartheta(F_nG_m)\;=\vartheta(F)\vartheta(G)\;=\;\lim_{n\to\infty}\lim_{m\to\infty} \vartheta(G_mF_n)\;=\;\vartheta(GF)\;.$$
		Finally, since it is multiplicative and $\mathfrak{K}$ is abelian, it follows that $\vartheta$ defines a pure state \cite[Corollary 2.3.21]{Bratelli}.
	\end{proof}

	\medskip

	Since $\mathfrak{K}$ is a $C^*$-subalgebra of $\mathfrak{A}$ one can invoke
	\cite[Proposition 2.3.24]{Bratelli} to define extensions of the state $\vartheta$ defined by \eqref{eq:def_eta} to the full algebra $\mathfrak{A}$. Interestingly, 
	any state resulting from such extension is a frustration-free ground state for the {HA-QDM dynamics} described in Section \ref{sec:dyn}.

	\begin{prop}\label{prop:Exixst}
		Let $\omega_0$ be any extension of the seed state $\vartheta$ defined by \eqref{eq:def_eta}. Then $\omega_0\in\mathtt{FG}_{\rm QDM}(\mathfrak{A})$ is a frustration-free ground state for  the {HA-QDM dynamics}.
	\end{prop}%
	\begin{proof}
		The conditions $\omega_0(A_tB_\gamma)=\vartheta(A_tB_\gamma)=1$  for every $t\in {\rm Hom}_0(\f{C},\n{G})^{-1}$ and $\gamma\in {\rm Hom}_0(\f{C},\n{G})_1$ imply after a simple calculation that the frustration-free condition holds. It only remains to prove that $\omega_0$ is a ground state. For that,
		let $A\in\mathfrak{A}_\Lambda$, with $\Lambda\subset\f{K}$ finite.
		By the conditions of the HA-QDM interaction, the sets
		\[
		\left\{
		x\in\f{K}\;\middle|\;
		{\rm supp}(A_x^0)\cap\Lambda\neq\emptyset
		\right\}
		\]
		and
		\[
		\left\{
		x\in\f{K}\;\middle|\;
		{\rm supp}(B_x^0)\cap\Lambda\neq\emptyset
		\right\}
		\]
		are finite.
		
		Therefore, there exists $n_0\in\n{N}$ such that, for every
		$n\geqslant n_0$, the finite set $\Lambda_n$ contains the supports
		of all the operators $A_x^0$ and $B_x^0$ whose supports intersect
		$\Lambda$. Moreover, if the support of an interaction term is
		disjoint from $\Lambda$, then that interaction term commutes with
		$A$. Consequently, for every $n\geqslant n_0$, one has
		
		\begin{align*}
			[H_{\rm QDM}^{\Lambda_n},A]
			&=
			-\sum_{\substack{x\in\f{K}\\
					{\rm supp}(A_x^0)\cap\Lambda\neq\emptyset}}
			[A_x^0,A]
			\\
			&\qquad
			-\sum_{\substack{x\in\f{K}\\
					{\rm supp}(B_x^0)\cap\Lambda\neq\emptyset}}
			[B_x^0,A].
		\end{align*}
		
		It follows that
		
		\begin{align*}
			-{\rm i}\omega_0(A^*\delta_{\rm QDM}(A))\;&=\;-{\rm i}\lim_{\Lambda_n\uparrow \f{K}}\omega_0(A^*{\rm i}[H_{\rm QDM}^{\Lambda_n},A])\\
			%&=\;\omega_0(A^*H_{\rm QDM}^\Lambda A-A^*AH_{\rm QDM}^\Lambda)\\
			&=\;\sum_{\substack{x\in \f{K}\\
					{\rm supp} (A_x^0)\cap \Lambda\not=\emptyset }}\pa{\omega_0\pa{A^*AA_x^0}-\omega_0\pa{A^*A_x^0 A}}\\
			&\qquad+ \sum_{\substack{x\in \f{K}\\
					{\rm supp} (B_x^0)\cap \Lambda\not=\emptyset }} \pa{\omega_0\pa{A^*AB_x^0}-\omega_0\pa{A^*B_x^0 A}}\;.
		\end{align*}
		By using 	Lemma \ref{lema estado}, one gets
		\begin{align*}
			-{\rm i}\omega_0(A^*\delta_{\rm QDM}(A))\;&=\;
			\sum_{\substack{x\in \f{K}\\
					{\rm supp} (A_x^0)\cap \Lambda\not=\emptyset }}\pa{\omega_0\pa{A^*A}-\omega_0\pa{A^*A_x^0 A}}\\
			&\qquad +\sum_{\substack{x\in \f{K}\\
					{\rm supp} (B_x^0)\cap \Lambda\not=\emptyset }}\pa{\omega_0\pa{A^*A}-\omega_0\pa{A^*B_x^0 A}}\\
			&=\sum_{\substack{x\in \f{K}\\
					{\rm supp} (A_x^0)\cap \Lambda\not=\emptyset }}\omega_0\pa{A^*\pa{\n{1}-A_x^0}A}\\
			&\qquad +\sum_{\substack{x\in \f{K}\\
					{\rm supp} (B_x^0)\cap \Lambda\not=\emptyset }}\omega_0\pa{A^*\pa{\n{1}-B_x^0}A}\;.
		\end{align*}
		Since $A_0^x$ is an  orthogonal projection, one has that
		\begin{equation*}
			\omega_0(A^*(\n{1}-A_x^0)A)\;=\;\omega_0(A^*(\n{1}-A_x^0)^*(\n{1}-A_x^0)A)\;\geqslant\; 0
		\end{equation*}
		Analogously, one has that $\omega_0(A^*({\n{1}-B_x^0})A)\geqslant 0$.
		In conclusion, one obtains the inequality $-{\rm i}\omega_0(A^*\delta_{\rm QDM}(A))\geqslant 0$ as a sum of nonnegative real numbers. This shows that $\omega_0$ meets the condition for the ground state of the  {HA-QDM dynamics}.
	\end{proof}

	\medskip

	One can summarize the main achievement of this section in the following result
	which justifies the name of seed state of the {HA-QDM dynamics} for $\vartheta$.

	\begin{cor}\label{cor_summ}
		One has that $\mathtt{FG}_{\rm QDM}(\mathfrak{A})\neq \emptyset$ and $\mathtt{FG}_{\rm QDM}^{{\circ}}(\mathfrak{A})\neq \emptyset$. Moreover 
		$\omega_0\in \mathtt{FG}_{\rm QDM}(\mathfrak{A})$  if and only if 
		$\omega_0|_{\mathfrak{K}}=\vartheta$, or equivalently if and only if 
		\begin{equation*}
			\omega_0(A_tB_\gamma)\;=\;1\;, \qquad \forall \;(t,\gamma)\in {\rm Hom}_0(\f{C},\n{G})^{-1}\times {\rm Hom}_0(\f{C},\n{G})_1\;.
		\end{equation*}
	\end{cor}

	\begin{proof}
		Proposition \ref{prop:Exixst} shows that $\mathtt{FG}_{\rm QDM}(\mathfrak{A})\neq \emptyset$ and in view of Proposition \ref{prop:frist_cond} any element $\omega_0\in \mathtt{FG}_{\rm QDM}(\mathfrak{A})$ is an extension of the pure state $\vartheta\in\mathtt{P}(\mathfrak{K})$ constructed in Lemma \ref{lemm:pur_st_K}. Finally, since $\vartheta$ is pure, it has at least one pure extension \cite[Proposition 2.3.24]{Bratelli} showing that 
		$\mathtt{FG}_{\rm QDM}^{{\circ}}(\mathfrak{A})\neq \emptyset$.
	\end{proof}

	\subsection{Characterization of the space of frustration-free ground states}
	In this section, we will introduce a middle $C^*$-algebra 
	$\mathfrak{K}\subset \mathfrak{K}'\subset \mathfrak{A}$
	which has the role of adding the necessary information 
	for differentiating between different extensions of the seed state $\vartheta$. More precisely, will we see that any $\tau\in \mathtt{S}(\mathfrak{K}')$ such that $\tau|_{\mathfrak{K}}=\vartheta$
	fix a unique $\omega_{0,\tau}\in \mathtt{FG}_{\rm QDM}(\mathfrak{A})$ by extension. For this reason, $\mathfrak{K}'$ will be called the \emph{fixing} $C^*$-algebra.

	\medskip

	Let $\omega_0\in \mathtt{FG}_{\rm QDM}(\mathfrak{A})$ be a given frustration-free ground state of the  {HA-QDM dynamics}. We are interested here in analyzing how $\omega_0$ acts on the monomial
	$P_tQ_\gamma$
	with $t\in {\rm Hom}_0(\f{C},\n{G})^0$ and $\gamma\in {\rm Hom}_0(\f{C},\n{G})_0$.
	
	\medskip  
	
	First of all, it is worth observing that, in view of Corollary \ref{cor_summ}, one has that
	\[
	\omega_0(P_{\delta^{-1}t'}Q_{\delta_1\gamma'})=\omega_0(A_{t'}B_{\gamma'})=1
	\]
	for every $t'\in {\rm Hom}_0(\f{C},\n{G})^{-1}$ and $\gamma'\in {\rm Hom}_0(\f{C},\n{G})_1$. One then gets
	\begin{equation}
		\omega_0(P_{t}Q_{\gamma})\;=\;1 \qquad \forall \;(t,\gamma)\in {\rm Im}{(\delta^{-1})}\times  {\rm Im}{(\delta_{1})}\;.
	\end{equation}
	Now, let $s\in {\rm Hom}_0(\f{C},\n{G})^{-1}$
	Using Lemma \ref{lema estado} and the commutation relations \eqref{relatP-Q}, one  gets
	\begin{align*}
		\omega_0(P_tQ_\gamma)\;&=\;\omega_0(A_s)\omega_0(P_tQ_\gamma)\omega_0(A_{-s})    \;=\;\omega_0(A_s P_t Q_\gamma A_{-s})\\
		&=\;\omega_0(P_{\delta^{-1}s} P_t Q_\gamma P_{-\delta^{-1}s})\;=\;\gamma(-\delta^{-1}s)\omega_0(P_tQ_\gamma)\\
		&=(\delta_0\gamma)(-s)\omega_0(P_tQ_\gamma)
	\end{align*}
	where the last equality is implied by the duality in Lemma \ref{lem:duality}. If $\gamma\notin {\rm Ker}(\delta_0)$
	then $(\delta_0\gamma)(s)\neq 1$ for some $s\in {\rm Hom}_0(\f{C},\n{G})^{-1}$. Since the equality above must be valid for every $s\in {\rm Hom}_0(\f{C},\n{G})^{-1}$ one infers that $\omega_0(P_{t}Q_{\gamma})=0$ if $\gamma\notin {\rm Ker}(\delta_0)$.
	By replacing in the computation above $A_s$ with $B_\nu$, for some $\nu\in {\rm Hom}_0(\f{C},\n{G})_{1}$,  one ends with the equality
	$\omega_0(P_tQ_\gamma)=\nu(\delta^0t)\omega_0(P_tQ_\gamma)$.
	With the same argument above one gets $\omega_0(P_{t}Q_{\gamma})=0$ if $t\notin {\rm Ker}(\delta^0)$. To sum up one has that
	\begin{equation}
		\omega_0(P_{t}Q_{\gamma})\;=\;0 \qquad {\rm if}\; \; t\in{\rm supp}(\delta^0)\;\;\text{or}\;\;
		\gamma\in{\rm supp}(\delta_0)\;,
	\end{equation}
	where ${\rm supp}(\delta^0):={\rm Hom}_0(\f{C},\n{G})^0\setminus{\rm Ker}(\delta^0)$
	and ${\rm supp}(\delta_0):={\rm Hom}_0(\f{C},\n{G})_0\setminus{\rm Ker}(\delta_0)$. 
	Furthermore, using again Lemma \ref{lema estado}, one ends with
	\begin{equation}\label{eq:inv_class}
		\omega_0(P_tQ_\gamma)\;=\;\omega_0(A_sP_tQ_\gamma B_\nu)\;=\;\omega_0(P_{t+\delta^{-1}s}Q_{\gamma+\delta_1 \nu})\;.
	\end{equation}
	This shows that the values of $\omega_0(P_tQ_\gamma)$ depends only on 
	$t\in {\rm Hom}_0(\f{C},\n{G})^0$ modulo ${\rm Im}{(\delta^{-1})}$,
	and on $\gamma\in {\rm Hom}_0(\f{C},\n{G})_0$ modulo ${\rm Im}{(\delta_{1})}$.

	\medskip

	The analysis above fixes many of the values of  $\omega_0(P_tQ_\gamma)$ but one still has  a certain degree of freedom to completely specify $\omega_0$. More precisely to fix $\omega_0$ one has to know the numbers
	\begin{equation}\label{eq:values}
		c_{t,\gamma}\;:=\;\omega_0(P_tQ_\gamma)\;,\quad \forall\; (t,\gamma)\in\big( {\rm Ker}(\delta^0)\times {\rm Ker}(\delta_0)\big)\setminus \big({\rm Im}{(\delta^{-1})}\times  {\rm Im}{(\delta_{1})}\big)\;.
	\end{equation}
	In the following, we will describe how to parametrize the freedom provided by the latter conditions.

	\medskip

	It has emerged that a special role is played by the set of operators
	\[
	\mathfrak{Z}\;:=\;\left\{P_t Q_\gamma \;\middle|\; (t,\gamma)\in {\rm Ker}(\delta^0)\times {\rm Ker}(\delta_0)\right\}\;\subset\;\mathfrak{A}_{\rm loc}\;.
	\]
	We will refer to $\mathfrak{Z}$ as the set of \emph{units} for the 
	frustration-free ground states of the {HA-QDM dynamics}. 
	There is a deep relationship between the $C^*$-algebra $\mathfrak{K}$ in Definition \ref{def:K} and set
	$\mathfrak{Z}$. Let $\mathfrak{K}'$ be the \emph{commutant} of $\mathfrak{K}$ in $\mathfrak{A}$, \emph{i.e.} the set of all operators in $\mathfrak{A}$ that commute with every element of $\mathfrak{K}$. 
	Since $\mathfrak{K}$ is abelian one has $\mathfrak{K}\subseteq\mathfrak{K}'$. 
	As before, it is only necessary to consider the elements that have finite support.

	\begin{lemma}\label{lemma:commutant_loc}
		It holds true that
		\[
		\mathfrak{K}' 
		\;=\;  \overline{\mathfrak{A}_{\rm loc}\cap\mathfrak{K}'}^{\;\|\;\|}
		\;.
		\]
	\end{lemma}
	\begin{proof}
		Given an element $A\in\mathfrak{K}'$, we will construct a sequence of operators in $\mathfrak{A}_{\rm loc} \cap \mathfrak{K}'$ that converges to $A$. Let $n\in\NN$. By density of the local operators, there exists a $A_n\in\mathfrak{A}_{\rm loc}$ such that $$\norm{A-A_n}\;<\; \frac{1}{n}\;.$$
		Denote by $\Lambda$ the support of $A_n$ and consider the equivalence relation $\sim$ on $ {\rm Hom}_0(\f{C},\n{G})^{-1}$ by $r\sim s$ if $(\delta^{-1}s)|_\Lambda=(\delta^{-1}r)|_\Lambda$. Also consider the equivalence relation $\sim$ on ${\rm Hom}_0(\f{C},\n{G})_1$ by $\alpha\sim \beta$ if $(\delta_{1}\alpha)|_\Lambda=(\delta_{1}\beta)|_\Lambda$. With this, we define 
		$$
		A_n'\;:=\;\frac{1}{\mathcal{N}_n}\sum_{([s],[\alpha])\in \mathcal{I}_n }A_sB_\alpha A_n B_{-\alpha}A_{-s}$$
		where $\mathcal{I}_n={\rm Hom}_0(\f{C},\n{G})^{-1}/\sim\times{\rm Hom}_0(\f{C},\n{G})_1/\sim$ and $\mathcal{N}_n:=|\mathcal{I}_n|$ its cardinality. 
		Note that $A_n'$ is independent of the representative of the different classes. To see this, notice that $A_n'|_{\Lambda^c}=\n{1}$ and on $\Lambda$ we have that $(A_sB_\alpha A_n B_{-\alpha}A_{-s})|_\Lambda=(A_rB_\beta A_n B_{-\beta}A_{-r})_\Lambda$ by the construction of the equivalence relations.
		One also has that  $A_n'\in\mathfrak{A}_{\rm loc}$ since it is a finite sum of local elements. Also, this operator commutes with all the monomials $A_rB_\beta$. 
		In fact, a direct computation shows that
		\begin{align*}
			A_rB_{\beta}A_n'B_{-\beta}A_{-r}\;&=\;\frac{1}{\mathcal{N}_n}\sum_{([s],[\alpha])\in \mathcal{I}_n  }A_rB_\beta A_sB_\alpha A_n B_{-\alpha}A_{-s} B_{-\beta}A_{-r}\\
			&=\;\frac{1}{\mathcal{N}_n}\sum_{([s],[\alpha])\in \mathcal{I}_n  }A_{r+s}B_{\alpha+\beta}A_nB_{-\alpha-\beta}A_{-r-s}\\
		\end{align*}
		where one uses the commutativity between $A_{\pm r}$ and $B_{\pm\alpha}$. Introducing the new variables      
		$s':=s+r$ and $\alpha':=\alpha+\beta$ one gets 
		\begin{align*}
			A_rB_{\beta}A_n'B_{-\beta}A_{-r}\;&=\;\frac{1}{\mathcal{N}_n}\sum_{([s'-r],[\alpha'-\beta])\in \mathcal{I}_n  }
			A_{s'}B_{\alpha'}A_nB_{-\alpha'}A_{-s'}\\
			&=\:\frac{1}{\mathcal{N}_n}\sum_{([s'],[\alpha'])\in \mathcal{I}_n  }A_{s'}B_{\alpha'}A_nB_{-\alpha'}A_{-s'}\\
			&=\;A_n'\\
		\end{align*} 
		where the second equality is justified by the fact that we are adding all the possible equivalence classes. 
		By a density argument, one infers that 
		$A_n'\in \mathfrak{A}_{\rm loc}\cap \mathfrak{K}'$. Now, using the assumption $A\in\mathfrak{K}'$ one can write
		\begin{align*}
			\norm{A_n'-A}\;&=\;\norm{A_n'-\frac{1}{\mathcal{N}_n}\sum_{([s],[\alpha])\in \mathcal{I}_n  }A_sB_\alpha A B_{-\alpha}A_{-s}}\\
			&=\;\frac{1}{\mathcal{N}_n}\norm{\sum_{([s],[\alpha])\in \mathcal{I}_n  }A_sB_\alpha (A_n-A) B_{-\alpha}A_{-s}}\\
			&\leqslant\;\norm{A_n-A}\;<\;\frac{1}{n}
		\end{align*}
		where the first inequality is obtained by the triangle inequality. Therefore, $A_n'$ provides the desired sequence converging to $A$.
	\end{proof}

	\medskip

	We will refer to $\mathfrak{K}'$ as the fixing $C^*$-algebra
	for the  {HA-QDM dynamics} for reasons that will be clarified below.
	
	\medskip
	
	Let $\mathrm{span}(\mathfrak{Z})$ be the set of finite linear combinations of elements of $\mathfrak{Z}$. 
	\begin{lemma}\label{Lem:span-K'}
		It holds true that
		\[
		\mathfrak{K}'\;=\;\overline{\mathrm{span}(\mathfrak{Z})}^{\;\|\;\|}\;.
		\]
	\end{lemma}
	\begin{proof}
		Let $P_tQ_\gamma\in \mathfrak{Z}$ and $A_s\in\mathfrak{K}$ with  $s\in {\rm Hom}_0(\f{C},\n{G})^{-1}$. Then
		\begin{equation}\label{eq:aux_1}
			\begin{aligned}
				(P_tQ_\gamma)A_s\;&=\;P_tQ_\gamma P_{\delta^{-1}s}\;=\;\gamma(\delta^{-1}s)P_tP_{\delta^{-1}s}Q_\gamma\\
				&=\;(\delta_0\gamma)(s)P_{\delta^{-1}s}(P_tQ_\gamma)\;=\;A_s(P_tQ_\gamma)\;
			\end{aligned}
		\end{equation}
		where $(\delta_0\gamma)(s)=1$ for every $s$ in view of  the condition $\gamma\in{\rm Ker}(\delta_0)$. Therefore, $P_tQ_\gamma$ commutes with any $A_s$. In a similar way let  $B_\nu\in\mathfrak{K}$ with $\nu\in {\rm Hom}_0(\f{C},\n{G})_{1}$. Then
		\begin{equation}\label{eq:aux_2}
			\begin{aligned}
				(P_tQ_\gamma)B_\nu\;&=\;P_tQ_\gamma Q_{\delta_{1}\nu}\;=\;P_tQ_{\delta_{1}\nu}Q_\gamma
				\;=\;(\delta_{1}\nu)(t)Q_{\delta_{1}\nu}(P_tQ_\gamma)\\
				&=\;\nu(\delta^0 t)B_\nu(P_tQ_\gamma)\;=\;B_\nu(P_tQ_\gamma)
			\end{aligned}
		\end{equation}
		where again $\nu(\delta^0 t)=1$ for every $\nu$ since 
		$t\in{\rm Ker}(\delta^0)$. Then, $P_tQ_\gamma$ also commutes with any $B_\nu$. Then $P_tQ_\gamma$ commutes with any element in $\mathfrak{K}_0$. Let $T\in \mathfrak{K}$ and $\{T_n\}\subset \mathfrak{K}_0$ such that $T_n\to T$ in norm. Then $[T,P_tQ_\gamma]=[T-T_n,P_tQ_\gamma]$
		which implies that $\|[T,P_tQ_\gamma]\|\leqslant 2 \|T-T_n\|$. Since the right-hand side can be chosen arbitrarily small, one concludes that 
		$[T,P_tQ_\gamma]=0$.
		Therefore $P_tQ_\gamma\in \mathfrak{K}'$ and for the arbitrariness 
		of $t$ and $\gamma$ one concludes that $\mathfrak{Z}\subseteq \mathfrak{K}'$. Since $\mathfrak{K}'$ is closed under linear combinations and convergence in norm, one ends with  
		\[
		\overline{\mathrm{span}(\mathfrak{Z})}^{\;\|\;\|}\;\subseteq \;\mathfrak{K}'\;.
		\]
		For the opposite inclusion, consider a finite sum
		\[
		R\;=\;\sum_{i=1}^n\lambda_{i}P_{t_i}Q_{\gamma_i}\;\in\;\mathfrak{A}_{\rm loc}\cap \mathfrak{K}'
		\]
		with $t_i\in {\rm Hom}_0(\f{C},\n{G})^0$ and $\gamma_i\in {\rm Hom}_0(\f{C},\n{G})_0$ for every $i=1,\ldots,n$. The condition $R\in \mathfrak{K}'$ implies the commutation relations $[R,A_s]=0=[R,B_\nu]$
		for every $s\in {\rm Hom}_0(\f{C},\n{G})^{-1}$ and $\nu\in {\rm Hom}_0(\f{C},\n{G})_{1}$. The first of these conditions provides
		\[
		0\;=\;[R,A_s]\;=\;\sum_{i=1}^n\lambda_{i}[P_{t_i}Q_{\gamma_i},A_s]\;=\;
		\sum_{i=1}^n\lambda_{i}((\delta_0\gamma_i)(s)-1)P_{t_i+\delta^{-1}s}Q_{\gamma_i}
		\]
		after a computation similar to \eqref{eq:aux_1}. Since the monomials 
		$P_{t_i+\delta^{-1}s}Q_{\gamma_i}$ are linearly independents generators of $\mathfrak{A}_{\rm loc}$ (see Lemma \ref{lemm:lin_indip_g}) and the $\lambda_i\neq 0$ one concludes that $(\delta_0\gamma_i)(s)=1$. Since this should be valid for every $s$ one ends with $\{\gamma_i\}_{i=1,\ldots,n}\subset {\rm Ker}(\delta_0)$. Similarly one has that  
		\[
		0\;=\;[R,B_\nu]\;=\;\sum_{i=1}^n\lambda_{i}[P_{t_i}Q_{\gamma_i},B_\nu]\;=\;
		\sum_{i=1}^n\lambda_{i}(\nu(\delta^0t_i)-1)P_{t_i}Q_{\gamma_i+\delta_1\nu}
		\]
		following the same computation as in \eqref{eq:aux_2}. Again one obtains $\nu(\delta^0t_i)=1$ for every $\nu$ which implies that 
		$\{t_i\}_{i=1,\ldots,n}\subset {\rm Ker}(\delta^0)$. Then, one concludes that $
		\mathfrak{A}_{\rm loc}\cap\mathfrak{K}'\subseteq\mathrm{span}(\mathfrak{Z})$ which implies
		\[
		\overline{\mathrm{span}(\mathfrak{Z})}^{\;\|\;\|}\;\supseteq\;\overline{\mathfrak{A}_{\rm loc}\cap\mathfrak{K}'}^{\;\|\;\|}
		\;=\;\mathfrak{K}'\;.
		\] 
		For the last equality, one uses Lemma \ref{lemma:commutant_loc}.
	\end{proof}

	\medskip

	The next result describes the main role of the $C^*$-algebra 
	$\mathfrak{K}'$ in the study of $\mathtt{FG}_{\rm QDM}(\mathfrak{A})$,
	and justify the name of {fixing} $C^*$-algebra for it.
	\begin{cor}\label{cor_uni_ext}
		Let $\mathtt{S}_{\vartheta}(\mathfrak{K}')\subset \mathtt{S}(\mathfrak{K}')$ be the subset of states $\tau$ in $\mathfrak{K}'$ such that $\tau|_{\mathfrak{K}}=\vartheta$. Then, there is a bijection
		$ \mathtt{FG}_{\rm QDM}(\mathfrak{A})\simeq\mathtt{S}_{\vartheta}(\mathfrak{K}')$.
	\end{cor}
	\begin{proof}
		First, notice that a state $\tau$ in $\mathtt{S}_{\vartheta}(\mathfrak{K}')$ that satisfies $\tau|_{\mathfrak{K}}=\vartheta$ can be extended to a frustration-free ground state $\omega_0$ in $\mathtt{FG}_{\rm QDM}(\mathfrak{A})$ by the Proposition \ref{prop:Exixst}. Using Eq.\eqref{eq:values} we see that this extension is completely determined by the values it takes on operators of the form $P_tQ_\gamma$ with $(t,\gamma)\in {\rm Ker}(\delta^0)\times {\rm Ker}(\delta_0)$. But Lemma \ref{Lem:span-K'} tells us that these operators generate $\mathfrak{K}'$. Therefore, the values on operators $P_tQ_\gamma$ are already known, and we conclude that
		\begin{equation}\label{eq:extension}
			\omega_0(P_tQ_\gamma)\;=\;
			\left\{
			\begin{aligned}
				&\tau(P_tQ_\gamma)&&{\rm if}\;\;P_tQ_\gamma\in \mathfrak{K}'\;,\\
				&0&&{\rm if}\;\;P_tQ_\gamma\in \mathfrak{A}\setminus\mathfrak{K}'\;.
			\end{aligned}
			\right.
		\end{equation}
		So, one has that $\tau$ can be extended to a unique frustration-free ground state on $\mathfrak{A}$.
		On the other hand, let $\omega_0$ be a frustration-free ground state on $\mathfrak{A}$. Taking $\tau:=\omega_0|_{\mathfrak{K}'}$, this provides a state in $\mathfrak{K}'$ such that $\tau|_{\mathfrak{K}}=\vartheta$. Moreover, its extension to $\mathfrak{A}$ is equal to $\omega_0$.
		Therefore, there is a bijection between $\mathtt{FG}_{\rm QDM}(\mathfrak{A})$ and $\mathtt{S}_{\vartheta}(\mathfrak{K}')$.
	\end{proof}

	\subsection{The $C^*$-algebra of logical operators}
	The $C^*$-subalgebra $\mathfrak{K}'$ contains sufficient information for the specification  of the  frustration-free ground states of the 
	{HA-QDM dynamics}. However $\mathtt{S}(\mathfrak{K}')$ is too large: it contains not only the ground states 
	of interest, but also additional states that differ only by redundant stabilizer information. To remove this redundancy, one introduces an ideal $\mathfrak{J}\subset \mathfrak{K}'$ (referred to as the \emph{redundancy ideal}) with the crucial property that $\omega_0:\mathfrak{J}\to \{0\}$ 
	for every $\omega_0\in \mathtt{FG}_{\rm QDM}(\mathfrak{A})$. As a consequence, the quotient $\mathfrak{K}'/\mathfrak{J}$ contains all the 
	information necessary to specify the frustration-free ground states, while eliminating the redundant degrees of freedom associated with the stabilizers.

	\medskip

	The \emph{redundancy}  space is defined by
	\begin{equation}\label{eq:relacion de equivalencia}
		\mathfrak{J}\;:=\;\overline{{\rm span}\{P_tQ_\gamma-P_{t'}Q_{\gamma'}\;|\; t\sim t'\;,\;\gamma\sim \gamma'\}}^{\;\|\;\|}
	\end{equation}
	where $t\sim t'$ means that $t,t'\in {\rm Hom}_0(\f{C},\n{G})^0$ are representatives of the same cohomology class, \emph{i.e.} $[t]=[t']\in\mathring{H}^0(\f{C},\n{G})$. Similarly, $\gamma\sim \gamma'$ means that
	$[\gamma]=[\gamma']\in\mathring{H}_0(\f{C},\n{G})$.
	More explicitly $ \mathfrak{J}$ is generated by linear combinations of elements of the form $P_tQ_\gamma-P_{t+\delta^{-1}s}Q_{\gamma+\delta_1\nu}$ for $(t,\gamma)\in {\rm Ker}(\delta^0)\times {\rm Ker}(\delta_0)$, $s\in {\rm Hom}_0(\f{C},\n{G})^{-1}$ and $\nu\in {\rm Hom}_0(\f{C},\n{G})_{1}$.

	\begin{lemma}
		The space $\mathfrak{J}$ is a self-adjoint bilateral ideal in $\mathfrak{K}'$.
	\end{lemma}

	\begin{proof}
		First of all $\mathfrak{J}$ is by definition a linear space 
		closed under the operation of taking the adjoint. Let $P_{r}Q_{\sigma}\in \mathfrak{K}'$ (which means $(r,\sigma)\in{\rm Ker}(\delta^0)\times {\rm Ker}(\delta_0)$) 
		and observe that 
		\begin{equation*}
			P_rQ_{\sigma}\pa{P_tQ_\gamma-P_{t'}Q_{\gamma'}}\;=\;\sigma(t)P_{r+t}Q_{\sigma+\gamma}-\sigma(t')P_{r+t'}Q_{\sigma+\gamma'}\;.
		\end{equation*}
		If $t\sim t'$ then $\sigma(t')=\sigma(t+\delta^{-1}s)=\sigma(t)(\delta_0\sigma)(s)=\sigma(t)$ and in turn 
		\begin{equation*}
			P_rQ_{\sigma}\pa{P_tQ_\gamma-P_{t'}Q_{\gamma'}}\;=\;\sigma(t)(P_{r+t}Q_{\sigma+\gamma}-P_{r+t'}Q_{\sigma+\gamma'})\;.
		\end{equation*}
		If also $\gamma\sim \gamma'$ one gets that $r+t\sim r+t'$ and $\sigma+\gamma\sim
		\sigma+\gamma'$ and one concludes that
		\begin{equation*}
			\pa{P_tQ_\gamma-P_{t'}Q_{\gamma'}}\in \mathfrak{J}\quad \Rightarrow\quad
			P_rQ_{\sigma}\pa{P_tQ_\gamma-P_{t'}Q_{\gamma'}}\in \mathfrak{J}\;.
		\end{equation*}
		Let $\mathfrak{J}_0:={{\rm span}\{P_tQ_\gamma-P_{t'}Q_{\gamma'}\;|\; t\sim t'\;,\;\gamma\sim \gamma'\}}$ and $J\in \mathfrak{J}_0$.
		Then by linearity the argument above shows that $P_rQ_{\sigma}J\in \mathfrak{J}_0$. Now let $J\in \mathfrak{J}$ and $\{J_n\}\subset\mathfrak{J}_0$ such that $J_n\to J$. Since $P_rQ_{\sigma}J_n\in \mathfrak{J}_0$ and $P_rQ_{\sigma}J_n\to P_rQ_{\sigma}J$ one gets that 
		$P_rQ_{\sigma}J\in \mathfrak{J}$. Let $K\in\mathfrak{K}'$ and $\{K_m\}\subset \mathrm{span}(\mathfrak{Z})$ such that $K_m\to K$. If $J\in \mathfrak{J}$, then by linearity one gets that $K_m J\in \mathfrak{J}$ for every $m$. Moreover $K_m J\to KJ$ showing that also $KJ \in \mathfrak{J}$. This shows that $\mathfrak{J}$ is a closed left-ideal of $\mathfrak{K}'$. However, since $\mathfrak{J}$ is self-adjoint, it is also a right-ideal.
	\end{proof}

	\medskip

	The main property of the redundancy ideal is described in the following result.

	\begin{lemma}\label{lemma:ker_J}
		Let $\omega_0\in \mathtt{FG}_{\rm QDM}(\mathfrak{A})$ be a 
		frustration-free ground state for  the {HA-QDM dynamics}. Then 
		$\omega_0(J)=0$ for every $J\in \mathfrak{J}$.
	\end{lemma}

	\begin{proof}
		The covariance condition \eqref{eq:inv_class} can be restated as
		$\omega_{0}(P_tQ_\gamma-P_{t'}Q_{\gamma'})=0$
		if $t\sim t'$ and $\gamma\sim\gamma'$. The proof is completed by linearity and continuity of $\omega_0$.
	\end{proof}

	\medskip

	We are now in a position to introduce the $C^*$-algebra of \emph{logical (quantum gates) operators} $\mathfrak{A}_{\rm log}$ which will play a key role in the characterization of the space of frustration-free ground states. For that, we need to recall that the quotient of a $C^*$-algebra by a closed two-sided ideal is again a   $C^*$-algebra \cite[Theorem 3.1.4]{Murphy}. 
	\begin{definition}[Logical operators]
		Let  $\mathfrak{A}_{\rm log}:=\mathfrak{K}'/\mathfrak{J}$.
	\end{definition}

	\medskip

	The relevance of $\mathfrak{A}_{\rm log}$ for the description of the frustration-free ground states for the {HA-QDM dynamics} is described in the following result. Let us recall that the natural topology for spaces of states is the $\ast$-weak topology.

	\begin{teo}\label{theo:one_on_one_bijection}
		There is a homeomorphism $\Phi:\mathtt{S}(\mathfrak{A}_{\rm log})\to\mathtt{FG}_{\rm QDM}(\mathfrak{A})$ between the state space of $\mathfrak{A}_{\rm log}$ and the
		space of frustration-free ground states for the {HA-QDM dynamics}.
	\end{teo}
	\begin{proof}
		Let $\pi:\mathfrak{K}'\to \mathfrak{A}_{\rm log}$ be the canonical 
		projection which is a surjective $\ast$-homomorphism of $C^*$-algebras.
		Let $s\in {\rm Hom}_0(\f{C},\n{G})^{-1}$ and $\nu\in {\rm Hom}_0(\f{C},\n{G})_{1}$ and consider $A_s B_\nu=P_{\delta^{-1}s}Q_{\delta_1\nu}=\n{1}-(\n{1}-P_{\delta^{-1}s}Q_{\delta_1\nu})$. Since
		$(\n{1}-P_{\delta^{-1}s}Q_{\delta_1\nu})\in \mathfrak{J}$ it follows that $\pi(A_s B_\nu)=[\n{1}]\in \mathfrak{A}_{\rm log}$. Since this holds true for all $s$ and $\nu$ one infers that $\pi(\mathfrak{K})=\n{C}[\n{1}]\subset\mathfrak{A}_{\rm log}$.  
		For ${\rho}\in \mathtt{S}(\mathfrak{A}_{\rm log})$ let  
		$\tau_\rho:=\rho\circ \pi$. It turns out that  $\tau_\rho\in\mathtt{S}(\mathfrak{K}')$. Moreover
		\[
		\tau_\rho(A_s B_\nu)\;=\;\rho(\pi(A_s B_\nu))\;=\;\rho([\n{1}])\;=\; 1
		\]
		for every $s\in {\rm Hom}_0(\f{C},\n{G})^{-1}$ and $\nu\in {\rm Hom}_0(\f{C},\n{G})_{1}$. Therefore, $\tau_\rho|_{\mathfrak{K}}=\vartheta$
		in view of Lemma \ref{lemm:pur_st_K} and in turn
		$\tau_\rho$ has a unique extension to a $\omega_{0,\rho}\in \mathtt{FG}_{\rm QDM}(\mathfrak{A})$ in view of Corollary \ref{cor_uni_ext}.
		Then, one has a map $\Phi:\mathtt{S}(\mathfrak{A}_{\rm log})\to \mathtt{FG}_{\rm QDM}(\mathfrak{A})$ defined by $\Phi(\rho):=\omega_{0,\rho}$. If $\Phi(\rho)=\Phi(\rho')$ then by definition
		$\rho(\pi(R))=\rho'(\pi(R))$ for every $R\in \mathfrak{K}'$, or equivalently $\rho([R])=\rho'([R])$ for every $[R]\in \mathfrak{A}_{\rm log}$, implying $\rho=\rho'$. This shows that $\Phi$ is injective. 
		On the other hand, if $\omega_0\in \mathtt{FG}_{\rm QDM}(\mathfrak{A})$
		then $\omega_0(R+J)=\omega_0(R)$ for every $R\in \mathfrak{K}'$ and 
		$J\in \mathfrak{J}$ in view of Lemma \ref{lemma:ker_J}.
		Then, one can define a linear functional $\rho_{\omega_0}:\mathfrak{A}_{\rm log}\to\n{C}$ by $\rho_{\omega_0}([R]):=\omega_0(R)$ where $R\in\mathfrak{K}'$ is any representative of the class $[R]\in \mathfrak{A}_{\rm log}$. A straightforward computation shows that 
		$\rho_{\omega_0}$ is positive and normalized, \emph{i.e.} $\rho_{\omega_0}\in \mathtt{S}(\mathfrak{A}_{\rm log})$.
		Observing that $\Phi(\rho_{\omega_0})(R)=\rho_{\omega_0}(\pi(R))=\omega_0(R)$ for every $R\in \mathfrak{K}'$, one has that
		$\Phi(\rho_{\omega_0})|_{\mathfrak{K}'}= {\omega_0}|_{\mathfrak{K}'}$.
		Since ${\omega_0}|_{\mathfrak{K}'}$ identify uniquely the extension $\omega_0$ in view of Corollary \ref{cor_uni_ext}, one ends with $\Phi(\rho_{\omega_0})=\omega_0$. Then $\Phi$ is also surjective and the bijection is proved.
		Finally, let us prove that this bijection is a homeomorphism. To do this, consider a sequence of states $\{\rho_n\}_{n\in\NN}\subset \mathtt{S}(\mathfrak{A}_{\rm log})$ that converge to a state $\rho\in \mathtt{S}(\mathfrak{A}_{\rm log})$ in the $*$-weak topology. First, for an operator $A\in \mathfrak{K}'$ one has by equation \eqref{eq:extension} that
		$$(\Phi(\rho_n))(A)\;=\;\rho_n(\pi(A))\;\rightarrow\; \rho(\pi(A))\;=\;(\Phi(\rho))(A)\;.
		$$
		Now, for $A\in\mathfrak{A}_{\rm loc}$ write it as 
		$$A\;=\;\sum^m_{\substack{j=1 \\ (t_j,\gamma_j) \in {\rm Ker}(\delta^0) \times {\rm Ker}(\delta_0)}}P_{t_j}Q_{\gamma_j} + \sum^l_{\substack{j=1 \\ (t'_j,\gamma'_j) \not\in {\rm Ker}(\delta^0) \times {\rm Ker}(\delta_0)}}P_{t'_j}Q_{\gamma'_j}\;.$$
		Then,
		\begin{align*}
			(\Phi(\rho_n))(A)\;&=\;(\Phi(\rho_n))\pa{\sum^m_{\substack{j=1 \\ (t_j,\gamma_j) \in {\rm Ker}(\delta^0) \times {\rm Ker}(\delta_0)}} P_{t_j}Q_{\gamma_j}}\\
			&=\;\rho_n\pa{\pi\pa{\sum^m_{\substack{j=1 \\ (t_j,\gamma_j) \in {\rm Ker}(\delta^0) \times {\rm Ker}(\delta_0)}} P_{t_j}Q_{\gamma_j}}}
		\end{align*}
		Taking the limit,
		\begin{align*}
			\lim_{n\to\infty}(\Phi(\rho_n))(A)\;&=\;\rho\pa{\pi\pa{\sum^m_{\substack{j=1 \\ (t_j,\gamma_j) \in {\rm Ker}(\delta^0) \times {\rm Ker}(\delta_0)}} P_{t_j}Q_{\gamma_j}}}\;=\;(\Phi(\rho))(A)\;.
		\end{align*}
		Finally, take $A\in\mathfrak{A}$ and $\{A_n\}_{n\in\NN}\subset\mathfrak{A}_{\rm loc}$ with $A_n\to A$ in norm to conclude that  
		$\Phi(\rho_n)\to\Phi(\rho)$ in the $\ast$-weak topology. Since $\mathtt{S}(\mathfrak{A}_{\rm log})$ is compact and $\mathtt{FG}_{\rm QDM}(\mathfrak{A})$ is Hausdorff it follows that the continuous bijection $\Phi$ is automatically a homeomorphism.
	\end{proof}

	\subsection{Topological properties of the space of frustration-free ground states}
	The main message of Theorem \ref{theo:one_on_one_bijection} is that to distinguish between two different frustration-free ground states
	of the {HA-QDM dynamics}
	is enough to compare their images in $\mathfrak{A}_{\rm log}$. For this reason, it is relevant to study the state space $\mathtt{S}(\mathfrak{A}_{\rm log})$.

	\medskip

	If $t\sim t'$ and $\gamma\sim \gamma'$ then the monomial $P_tQ_\gamma$ and  $P_{t'}Q_{\gamma'}$ in $\mathfrak{K}'$ are sent by the canonical projection $\pi$ to the same class in $\mathfrak{A}_{\rm log}$, \emph{i.e.} $\pi(P_tQ_\gamma)=\pi(P_{t'}Q_{\gamma'})$. By denoting 
	with $[P_tQ_\gamma]$ the common class one deduces that the class only depends on $[t]\in \mathring{H}^0(\f{C},\n{G})$ and $[\gamma]\in\mathring{H}_{0}(\f{C},\n{G})$. This justifies the introduction of the notation $P_{[t]}Q_{[\gamma]}:=[P_tQ_\gamma]$. Let us also introduce the symbol
	\[
	\n{H}(\f{C},\n{G})\;:=\;\mathring{H}^0(\f{C},\n{G})\times\mathring{H}_0(\f{C},\n{G})\;.
	\]
	Since the projection $\pi$ is $\ast$-homomorphism, the commutation relations \eqref{relatP-Q} and \eqref{eq:commutation_P_and_Q} factor through the quotient and one obtains that
	\begin{equation}\label{eq:Weyl1}
		\pa{P_{[t]}Q_{[\gamma]}}\pa{P_{[t']}Q_{[\gamma']}}\;=\;\gamma(t')P_{[t+t']}Q_{[\gamma+\gamma']}
	\end{equation}
	for all $([t],[\gamma]),([t'],[\gamma'])\in \n{H}(\f{C},\n{G})$.
	%Moreover, from  Lemma \ref{lemm:lin_indip_g} and Lemma \ref{Lem:span-K'} one also infers that the set
	%\[
	%\mathfrak{Z}_\pi\;:=\;\pi(\mathfrak{Z})\;=\;\left\{P_{[t]}Q_{[\gamma]}\;|\;([t],[\gamma])\in\n{H}(\f{C},\n{G}) \right\}
	%\]
	%is linearly independent, and its linear span is norm dense in $\mathfrak{A}_{\rm log}$.

	\medskip

	Let us introduce the notation ${\bf h}:=([t],[\gamma])$ for a generic element of 
	$\n{H}(\f{C},\n{G})$ and
	the form
	\[
	\langle\;\cdot\;,\;\cdot\;\rangle\;:\;\n{H}(\f{C},\n{G})\times \n{H}(\f{C},\n{G})\;\longrightarrow\;\n{U}(1)
	\]
	defined by 
	\[
	\langle {\bf h} ,{\bf h}'\rangle\;:=\;\gamma(t')
	\]
	for all ${\bf h}=([t],[\gamma])$ and ${\bf h}'=([t'],[\gamma'])$ in 
	$\n{H}(\f{C},\n{G})$. One can check that 
	\[
	\begin{aligned}
		\langle {\bf h} +{\bf h}',{\bf h}''\rangle\;&=\;(\gamma+\gamma')(t'')\;=\;\gamma(t'')\gamma'(t'')\;=\;\langle {\bf h} ,{\bf h}''\rangle\langle{\bf h}',{\bf h}''\rangle\\
		\langle {\bf h} ,{\bf h}'+{\bf h}''\rangle\;&=\;\gamma(t'+t'')\;=\;\gamma(t')\gamma(t'')\;=\;\langle {\bf h} ,{\bf h}'\rangle\langle{\bf h},{\bf h}''\rangle
	\end{aligned}
	\]
	Introducing the notation $W_{{\bf h}}:=P_{[t]}Q_{[\gamma]}$ one can rephrase the relations \eqref{eq:Weyl1} as
	\begin{equation}\label{eq:Weyl2}
		W_{{\bf h}}W_{{\bf h}'}\;=\;\langle {\bf h} ,{\bf h}'\rangle\;W_{{\bf h}+{\bf h}'}\;.
	\end{equation}
	These operators completely encode the generators of $\mathfrak{A}_{\rm log}$, as shown by the following lemma.

	%%%%%%%%%%%%%%%%new proof%%%%%%%%%%%%%%%
	
	\begin{lemma}\label{lem:linear_independence_logical}
		The set
		$\left\{W_{\bf h}\;\middle|\;{\bf h}\in\n{H}(\f{C},\n{G})\right\}
		$
		is linearly independent, and its linear span is norm dense in $\mathfrak{A}_{\rm log}$.
	\end{lemma}
	
	\begin{proof}
		Consider the Hilbert space
		$
		\ell^2\pa{\n{H}(\f{C},\n{G})}
		$
		with its canonical orthonormal basis
		$
		\left\{\delta_{\bf k}\;\middle|\;{\bf k}\in\n{H}(\f{C},\n{G})\right\}.
		$
		For every ${\bf h}\in\n{H}(\f{C},\n{G})$, define the operator $V_{\bf h}$ by
		\begin{equation}\label{eq:regular_representation_H}
			V_{\bf h}\delta_{\bf k}\;:=\;\langle{\bf h},{\bf k}\rangle \delta_{{\bf h}+{\bf k}}.
		\end{equation}
		Since $\ab{\langle{\bf h},{\bf k}\rangle}=1$ and the map ${\bf k}\mapsto{\bf h}+{\bf k}$ is a bijection, $V_{\bf h}$ is a unitary operator.
		Moreover, using the properties of $\langle\,\cdot\,,\,\cdot\,\rangle$, for every ${\bf h},{\bf h}',{\bf k}\in\n{H}(\f{C},\n{G})$ one has
		\begin{align*}
			V_{\bf h}V_{{\bf h}'}\delta_{\bf k}\;=\;\langle{\bf h},{\bf h}'\rangle\langle{\bf h}',{\bf k}\rangle\langle{\bf h},{\bf k}\rangle\delta_{{\bf h}+{\bf h}'+{\bf k}}\;=\;\langle{\bf h},{\bf h}'\rangle V_{{\bf h}+{\bf h}'}\delta_{\bf k}\;.
		\end{align*}
		Hence
		\begin{equation}\label{eq:regular_Weyl_relation}
			V_{\bf h}V_{{\bf h}'}\;=\;\langle{\bf h},{\bf h}'\rangle V_{{\bf h}+{\bf h}'}\;.
		\end{equation}
		We now use these operators to construct a representation of $\mathfrak{K}'$. Let $\Lambda\subset\f{K}$ be finite and set
		$
		\mathfrak{K}'_\Lambda:=\mathfrak{A}_\Lambda\cap\mathfrak{K}'.
		$
		By Lemma~\ref{lemm:lin_indip_g} and the proof of Lemma~\ref{Lem:span-K'}, every $R\in\mathfrak{K}'_\Lambda$ admits a unique expression of the form
		$
		R=\sum_{j=1}^n c_j P_{t_j}Q_{\gamma_j}
		$,
		where
		$
		t_j\in{\rm Ker}(\delta^0)$, $\gamma_j\in{\rm Ker}(\delta_0),
		$
		and the operators $P_{t_j}Q_{\gamma_j}$ have support contained in $\Lambda$.
		We can therefore define a linear map
		\[
		\rho_\Lambda\;:\; \mathfrak{K}'_\Lambda\;\longrightarrow\;\mathcal{B}\pa{\ell^2\pa{\n{H}(\f{C},\n{G})}}
		\]
		by setting
		\begin{equation}\label{eq:rho_local_definition}
			\rho_\Lambda(P_tQ_\gamma)\;:=\;V_{([t],[\gamma])}
		\end{equation}
		on the generators and extending by linearity. The uniqueness of the above expansion shows that $\rho_\Lambda$ is well-defined.
		Let
		$
		{\bf h}=([t],[\gamma])$, ${\bf h}'=([t'],[\gamma'])
		$.
		Using the Weyl relations in $\mathfrak{K}'$, together with \eqref{eq:regular_Weyl_relation}, one obtains
		\begin{align*}
			\rho_\Lambda\pa{P_tQ_\gamma P_{t'}Q_{\gamma'}}\;&=\;\rho_\Lambda\pa{\gamma(t')P_{t+t'}Q_{\gamma+\gamma'}}\;=\;\langle{\bf h},{\bf h}'\rangle V_{{\bf h}+{\bf h}'}\\
			&=\;V_{\bf h}V_{{\bf h}'}\;
			=\;\rho_\Lambda(P_tQ_\gamma)\rho_\Lambda(P_{t'}Q_{\gamma'})\;.
		\end{align*}
		The adjoint relation is preserved as well. Indeed,
		\[
		(P_tQ_\gamma)^*\;=\;Q_{-\gamma}P_{-t}\;=\;\gamma(t)P_{-t}Q_{-\gamma}\;,
		\]
		while \eqref{eq:regular_Weyl_relation} gives
		$
		V_{\bf h}^*=\gamma(t)V_{-{\bf h}}
		$.
		Therefore
		$
		\rho_\Lambda((P_tQ_\gamma)^*)=\rho_\Lambda(P_tQ_\gamma)^*
		$.
		It follows that $\rho_\Lambda$ is a unital $\ast$-homomorphism.
		Since $\mathfrak{K}'_\Lambda$ is a $C^*$-algebra, $\rho_\Lambda$ is contractive. Thus
		\begin{equation}\label{eq:rho_local_contractive}
			\norm{\rho_\Lambda(R)}\;\leqslant\;\norm{R},\qquad R\in\mathfrak{K}'_\Lambda\;.
		\end{equation}
		The maps $\rho_\Lambda$ are compatible as $\Lambda$ increases. Indeed, if $\Lambda\subset\Lambda'$ and $R\in\mathfrak{K}'_\Lambda$, the uniqueness in Lemma~\ref{lemm:lin_indip_g} implies that the expansion of $R$ in terms of the operators $P_tQ_\gamma$ is the same whether it is regarded as an element of $\mathfrak{A}_\Lambda$ or of $\mathfrak{A}_{\Lambda'}$. Hence
		$
		\rho_{\Lambda'}(R)=\rho_\Lambda(R)
		$.
		Consequently, the family $\{\rho_\Lambda\}_\Lambda$ defines a $\ast$-homomorphism
		\[
		\rho_0\;:\;\mathfrak{A}_{\rm loc}\cap\mathfrak{K}'\;\longrightarrow\;\mathcal{B}\pa{\ell^2\pa{\n{H}(\f{C},\n{G})}}
		\]
		satisfying
		$
		\rho_0(P_tQ_\gamma)=V_{([t],[\gamma])}
		$.
		Moreover, \eqref{eq:rho_local_contractive} shows that
		$
		\norm{\rho_0(R)}\leqslant\norm{R}
		$
		for every $R\in\mathfrak{A}_{\rm loc}\cap\mathfrak{K}'$.
		By Lemma~\ref{lemma:commutant_loc}, $\mathfrak{A}_{\rm loc}\cap\mathfrak{K}'$ is norm dense in $\mathfrak{K}'$. Therefore, $\rho_0$ extends uniquely to a $\ast$-representation
		\begin{equation}\label{eq:representation_Kprime}
			\rho:\mathfrak{K}'\longrightarrow\mathcal{B}\pa{\ell^2\pa{\n{H}(\f{C},\n{G})}},
		\end{equation}
		which satisfies
		$
		\rho(P_tQ_\gamma)=V_{([t],[\gamma])}
		$.
		We now show that the representation factors through the redundancy ideal. Recall that $\mathfrak{J}$ is the norm closure of the linear span of the elements
		$
		P_tQ_\gamma-P_{t'}Q_{\gamma'}
		$
		with
		$
		t\sim t'$ and $\gamma\sim\gamma'
		$.
		The latter condition implies
		$
		([t],[\gamma])=([t'],[\gamma'])
		$,
		and therefore
		\begin{align*}
			\rho\pa{P_tQ_\gamma-P_{t'}Q_{\gamma'}}
			\;=\;V_{([t],[\gamma])}-V_{([t'],[\gamma'])}\;=\;0\;.
		\end{align*}
		By linearity and continuity of $\rho$, it follows that
		$
		\mathfrak{J}\subseteq{\rm Ker}(\rho)
		$.
		Let
		$
		\pi:\mathfrak{K}'\to\mathfrak{A}_{\rm log}=\mathfrak{K}'/\mathfrak{J}
		$
		be the canonical projection. Since $\mathfrak{J}\subseteq{\rm Ker}(\rho)$, the universal property of the quotient provides a unique $\ast$-representation
		\[
		\overline{\rho}:\mathfrak{A}_{\rm log}\longrightarrow\mathcal{B}\pa{\ell^2\pa{\n{H}(\f{C},\n{G})}}
		\]
		such that
		$
		\overline{\rho}\circ\pi=\rho
		$.
		In particular, for
		$
		{\bf h}=([t],[\gamma])
		$
		one has
		\begin{equation}\label{eq:logical_regular_representation}
			\overline{\rho}(W_{\bf h})\;=\;\overline{\rho}\pa{P_{[t]}Q_{[\gamma]}}\;=\;V_{\bf h}\;.
		\end{equation}
		Finally, suppose that for some finite subset $F\subset\n{H}(\f{C},\n{G})$ one has
		$
		\sum_{{\bf h}\in F}c_{\bf h}W_{\bf h}=0
		$.
		Applying $\overline{\rho}$ and then evaluating the resulting operator at $\delta_{\bf 0}$ gives
		\begin{align*}
			0\;&=\;\sum_{{\bf h}\in F}c_{\bf h}V_{\bf h}\delta_{\bf 0}\;
			=\;\sum_{{\bf h}\in F}c_{\bf h}\langle{\bf h},{\bf 0}\rangle\delta_{\bf h}\;
			=\;\sum_{{\bf h}\in F}c_{\bf h}\delta_{\bf h}\;.
		\end{align*}
		Since the vectors $\{\delta_{\bf h}\}_{{\bf h}\in F}$ are linearly independent, it follows that $c_{\bf h}=0$ for every ${\bf h}\in F$.
		Therefore, the operators
		$
		\left\{W_{\bf h}\;\middle|\;{\bf h}\in\n{H}(\f{C},\n{G})\right\}
		$
		are linearly independent.
		Let us now prove that the linear span of the operators $W_{\bf h}$ is norm dense in $\mathfrak{A}_{\rm log}$.
		By Lemma~\ref{Lem:span-K'}, one has
		\[
		\mathfrak{K}'=\overline{{\rm span}\left\{P_tQ_\gamma\;\middle|\;t\in{\rm Ker}(\delta^0),\ \gamma\in{\rm Ker}(\delta_0)\right\}}^{\|\cdot\|}\;.
		\]
		Since the canonical projection
		$
		\pi:\mathfrak{K}'\longrightarrow\mathfrak{A}_{\rm log}
		$
		is continuous and surjective, it follows that
		\begin{align*}
			\mathfrak{A}_{\rm log}&\;=\;\pi(\mathfrak{K}')\\
			&\subseteq\;\overline{\pi\pa{{\rm span}\left\{P_tQ_\gamma\;\middle|\;t\in{\rm Ker}(\delta^0),\ \gamma\in{\rm Ker}(\delta_0)\right\}}}^{\|\cdot\|}\\
			&=\;\overline{{\rm span}\left\{W_{\bf h}\;\middle|\;{\bf h}\in\n{H}(\f{C},\n{G})\right\}}^{\|\cdot\|}\;.
		\end{align*}
		The opposite inclusion is immediate. Therefore,
		\[
		\mathfrak{A}_{\rm log}\;=\;\overline{{\rm span}\left\{W_{\bf h}\;\middle|\;{\bf h}\in\n{H}(\f{C},\n{G})\right\}}^{\|\cdot\|}
		\]
		and this concludes the proof.	
	\end{proof}
	
	%%%%%%%%%%%%%%%%up to here%%%%%%%%%%%%

	The relations \eqref{eq:Weyl2} are generalized CCR relations in Weyl form \cite[Section 5.2.2.2]{bratelli2}. In view of the previous lemma, $\mathfrak{A}_{\rm log}$ is therefore a (generalized) CCR $C^*$-algebra.

	\medskip

	If
	one assumes that  $\n{H}(\f{C},\n{G})$ is finite then one ends in the same setting of  \cite[Lemma 6.4]{Vrana}.
	For that let us introduce the commutator $[[{\bf h},{\bf h}']]:=\langle {\bf h} ,{\bf h}'\rangle\langle {\bf h}' ,{\bf h}\rangle^{-1}$ and the commutator subgroup
	\[
	\n{I}(\f{C},\n{G})\;:=\; \left\{{\bf h}\in \n{H}(\f{C},\n{G})\;\big|\;[[{\bf h},{\bf h}']]=1\;,\;\;\forall\; {\bf h}'\in  \n{H}(\f{C},\n{G})\right\}\;.
	\]
	Let us define the numbers
	\[
	c\;:=\; \log_2(|\n{I}(\f{C},\n{G})|)\;,\qquad q\;:=\;\frac{1}{2}\log_2\left(\frac{|\n{H}(\f{C},\n{G})|}{|\n{I}(\f{C},\n{G})|}\right)
	\]
	the finite set $X_c:=\{1,2,\ldots,2^c\}$ and the finite dimensional Hilbert space $\mathfrak{h}_q:=\n{C}^{2^q}$. With this notation, the content of \cite[Lemma 6.4]{Vrana} can be adapted to our situation providing the isomorphism
	\begin{equation}\label{eq:trivial bundle}
		\mathfrak{A}_{\rm log}\;\simeq\; C(X_c)\otimes \mathcal{B}(\mathfrak{h}_q)\;.
	\end{equation}
	Interestingly,  $\mathfrak{A}_{\rm log}$ is isomorphic to the tensor product of the commutative algebra $C(X_c)$ of \emph{classical} degrees of freedom and a non-commutative algebra $\mathcal{B}(\mathfrak{h}_q)$
	of \emph{quantum} degrees of freedom. 
	%The relevance of this result is that it provides a way to determine the number $c$ of classical bits and the number $q$ of qubits that one can encode in the frustration-free ground space. 

	\medskip

	\begin{remark}
		It is useful to provide a more explicit characterization for the element of  $\n{I}(\f{C},\n{G})$.
		Notice that  ${\bf h} =([t],[\gamma]) \in \n{I}(\f{C},\n{G})$
		if and only if 
		\begin{equation*}
			1\;=\;\langle {\bf h} ,{\bf h}'\rangle\langle {\bf h}' ,{\bf h}\rangle^{-1}=\gamma(t')\gamma'(t)^{-1}
		\end{equation*}
		for all ${\bf h'} =([t'],[\gamma'])$.
		As this happens for all ${\bf h'}$, taking $[t']$ as the trivial class on $\mathring{H}^0(\f{C},\n{G})$, one obtains that $\gamma'(t)^{-1}=1$ for all $[\gamma']$. Analogously, taking $[\gamma']$ as the trivial class on $\mathring{H}_0(\f{C},\n{G})$, one has $\gamma(t')=1$ for all $[t']$.
		In conclusion, ${\bf h} =([t],[\gamma]) \in \n{I}(\f{C},\n{G})$
		if and only if  $\alpha(t)=1=\gamma(s)$
		for all $[s]\in \mathring{H}^0(\f{C},\n{G})$ and $[\alpha]\in \mathring{H}_0(\f{C},\n{G})$.
		\hfill $\blacktriangleleft$
	\end{remark}

	\medskip

	\begin{remark}[Compact manifold]
		For models arising from triangulations of compact manifolds \cite{Lorca2025, dealmeida2017topological}, it is proved that the frustration-free ground states depend on the Cohomology group $H^0(\f{C},\n{G})$ and the Homology group $H_0(\f{C},\n{G})$. 
		Notice that in those cases, $\mathring{H}^0(\f{C},\n{G}) = H^0(\f{C},\n{G})$ and $\mathring{H}_0(\f{C},\n{G}) = H_0(\f{C},\n{G})$. Therefore, our model also includes the ones associated to compact manifolds.
		\hfill $\blacktriangleleft$
	\end{remark}

	\subsection{Description of the pure states space}
	A relevant payoff of the result obtained in the previous section is a complete description of the space  $\mathtt{FG}_{\rm QDM}^{\circ}(\mathfrak{A})$ of the pure frustration-free ground states of the HA-QDM dynamics. This is a relevant fact since the full space $\mathtt{FG}_{\rm QDM}(\mathfrak{A})$ can be recovered as the closure of the convex hull of $\mathtt{FG}_{\rm QDM}^{\circ}(\mathfrak{A})$.

	\medskip

	\begin{prop}\label{prop:bijection_pure_states}
		There is a homeomorphism $\Phi:\mathtt{P}(\mathfrak{A}_{\rm log})\to\mathtt{FG}^\circ_{\rm QDM}(\mathfrak{A})$ between the  space of pure states of $\mathfrak{A}_{\rm log}$ and the
		space of pure frustration-free ground states for the {HA-QDM dynamics}.
	\end{prop}

	\begin{proof}
		First, consider the case where $\rho$ is mixed. Then, there exist distinct states $\rho_1,\rho_2$ in $\mathtt{S}(\mathfrak{A}_{\rm log})$ and $\lambda\in (0,1)$ with
		$$\rho\;=\;\lambda\rho_1+(1-\lambda)\rho_2$$
		By the bijection, one concludes that $\omega_{0,\rho}$ is mixed.
		On the other hand, if $\omega_{0,\rho}$ is mixed, it can be written as 
		$$\omega_{0,\rho}=\lambda \omega_1+(1-\lambda)\omega_2$$
		where $\omega_1$ and $\omega_2$ are different states and $\lambda\in (0,1)$. If we prove that $\omega_1$ and $\omega_2$ are frustration-free ground states, then one concludes the argument.
		Take $A_sB_\nu$ with $s\in {\rm Hom}_0(\f{C},\n{G})^{-1}$ and $\nu\in {\rm Hom}_0(\f{C},\n{G})_{1}$. This operator is unitary and this implies that $|\omega_1(A_sB_\nu)|\leq1$ and $|\omega_2(A_sB_\nu)|\leq 1$. But,
		$$
		1\;=\;\omega_{0,\rho}(A_sB_\nu)\;=\;\lambda \omega_1(A_sB_\nu)+(1-\lambda)\omega_2(A_sB_\nu)$$
		and this is only possible if $\omega_1(A_sB_\nu)=\omega_2(A_sB_\nu)=1$. Then, $\omega_1$ and $\omega_2$ are frustration-free ground states. Using the bijection, we conclude that $\rho$ is also mixed. Finally, $\Phi$ 
		is the restriction of a homeomorphism, and is again a 
		homeomorphism with respect to the induced subspace topology
		in the initial and final space.
	\end{proof}

	\medskip

	By using the tensor product structure of $\mathfrak{A}_{\rm log}$ described by \eqref{eq:trivial bundle} and \cite[Proposition 2.9]{denittis2025} one ends with a homeomorphism of topological spaces
	$$
	\mathtt{FG}_{\rm QDM}^{\circ}(\mathfrak{A})\;\simeq\; X_c\times \mathtt{P}(\mathcal{B}(\mathfrak{h}_q))\;.
	$$
	Finally, since $\mathcal{B}(\mathfrak{h}_q)$ is a full matrix-algebra 
	on the finite dimensional Hilbert space $\mathfrak{h}_q$, its pure states are described by one-dimensional projections, or equivalently by one-dimensional subspaces. Therefore, one has the  homeomorphism of topological spaces
	$$
	\mathtt{FG}_{\rm QDM}^{\circ}(\mathfrak{A})\;\simeq\; X_c\times \n{P}(\mathfrak{h}_q)
	$$
	where $\n{P}(\mathfrak{h}_q)\simeq \n{CP}^{2^q}$ is the complex projective space in $\mathfrak{h}_q=\n{C}^{2^q}$.

	\section{$C^*$-diagonal structure}\label{sec:diagonal}

	As we have seen, the algebra $\mathfrak{K}$ is the one that allows us to get all the frustration-free ground states from the \textit{seed state}. One interesting question is, under which conditions, this algebra is a diagonal one. The reason for our interest is because this condition for planar models has been proven to be useful to classify the KMS-states (see \cite{PoloOjitoProdan2026}).
	
	\subsection{Basic definitions}

	A \(C^{*}\)-diagonal is a specific type of subalgebra inside a larger \(C^{*}\)-algebra, playing a role similar to the diagonal matrices within the algebra of all bounded operators. Formally, it is a \(C^{*}\)-subalgebra \(\mathfrak{D}\) of a \(C^{*}\)-algebra \(\mathfrak{A}\) that meets a certain list of conditions. The following definition is borrowed from \cite[Definition 5.1]{Renault2008}.
	\begin{definition}[\(C^{*}\)-diagonal]\label{def:cartan}
		Let $\mathfrak{A}$ be a unital $C^*$-algebra. Then $\mathfrak{D}\subseteq \mathfrak{A}$ is called a \emph{Cartan subalgebra} if
		\begin{enumerate}
			\item[(a)] $\mathfrak{D}$ contains the unit,
			\item[(b)] $\mathfrak{D}$ is a maximal abelian sub-$C^*$-algebra (MASA) of $\mathfrak{A}$,
			\item[(c)] $\mathfrak{A}$ is generated by 
			\[
			\mathcal{N}(\mathfrak{D}):=\set{V\in\mathfrak{A} ~\lvert~ VBV^*\in\mathfrak{D} \text{ and } V^*BV\in\mathfrak{D} \text{ for all }B\in\mathfrak{D}}
			\]
			in which case $\mathfrak{D}$ is called \emph{regular},
			\item[(d)] there exists a faithful conditional expectation $\n{E}:\mathfrak{A}\to\mathfrak{D}$.
		\end{enumerate}
		Moreover, $\mathfrak{D}\subseteq \mathfrak{A}$ is called a $C^*$-\emph{diagonal} if additionally
		\begin{enumerate}
			\item[(e)] every pure state on $\mathfrak{D}$ has a unique extension to a pure state on $\mathfrak{A}$.
		\end{enumerate}
	\end{definition}
	
	\medskip
	
	Condition (e) is commonly called the \emph{Unique Extension Property} (UEP). The separability of $\mathfrak{A}$ simplifies the definition above, since several of these properties turn out to be equivalent.

	\begin{lemma}[Separable case]\label{lemm:Sep}
		Let $\mathfrak{A}$ be separable and unital. If $\mathfrak{D}\subseteq \mathfrak{A}$ is abelian, is regular, and has the UEP, then $\mathfrak{D}$ is a $C^*$-diagonal of $\mathfrak{A}$.
	\end{lemma}
	\begin{proof}
		Since $\mathfrak{D}$ has the UEP, \cite[Corollary 2.7]{Archbold} guarantees that $\mathfrak{D}$ is a maximal abelian sub-$C^*$-algebra of $\mathfrak{A}$.
		Every $\omega\in\mathtt{P}(\mathfrak{D})$ has, by hypothesis, a unique pure extension $s_\omega\in \mathtt{P}(\mathfrak{A})$. Given $A\in\mathfrak{A}$, let $f_A:\mathtt{P}(\mathfrak{D})\to\n{C}$ be the function uniquely specified by $f_A(\omega):=s_\omega(A)$. A direct check of continuity shows that $f_A\in C(\mathtt{P}(\mathfrak{D}))\simeq \mathfrak{D}$, where the isomorphism is given by the Gelfand transform $\iota$. This induces a map $\n{E}:\mathfrak{A}\to\mathfrak{D}$, defined by $\n{E}(A):=\iota(f_A)$, which is a conditional expectation. The uniqueness of such a conditional expectation is exactly the content of \cite[Theorem 3.4]{Anderson}.
	\end{proof}

	\subsection{Criterion for the $C^*$-diagonality}
	The aim of this section is to establish a topological criterion characterizing when $\mathfrak{K}$ is a $C^*$-diagonal in $\mathfrak{A}$. We start by collecting some preliminary facts.
	First of all, let us recall that $\mathfrak{A}$ is unital and separable, since it is generated by a countable family as described in Section \ref{sec:Al-Op}. Moreover, by Definition \ref{def:K}, one has that $\mathfrak{K}\subset\mathfrak{A}$ contains the unit and is abelian. Therefore, we are in the setting of Lemma \ref{lemm:Sep}.
	
	\begin{lemma}\label{prop:regular}
		$\mathfrak{K}$ is a regular sub-algebra of $\mathfrak{A}$.
	\end{lemma}
	\begin{proof}
		Take $A_sB_\xi\in \mathfrak{K}$. By the commutation relation \eqref{eq:commutation_P_and_Q} we obtain
		
		\begin{equation}
			P_{t}Q_\gamma A_sB_\xi Q_{\gamma}^*P_{t}^*\;=\;c A_sB_\xi\;\in\; \mathfrak{K}\;
		\end{equation}
		where $c$ is a complex number. Analogously, $(P_tQ_\gamma)^*A_sB_\xi P_tQ_\gamma\in\mathfrak{K}$.
		By linearity, for any element $T\in \mathfrak{K}_0$ we have $P_tQ_\gamma T(P_tQ_\gamma)^*\in \mathfrak{K}$ and $(P_tQ_\gamma)^*TP_tQ_\gamma\in \mathfrak{K}$. 
		By continuity and density also $P_tQ_\gamma K(P_tQ_\gamma)^*\in \mathfrak{K}$ and $(P_tQ_\gamma)^*KP_tQ_\gamma\in \mathfrak{K}$ for any $K\in\mathfrak{K}$.	 Therefore, any element of the form $P_tQ_\gamma$ is in $\mathcal{N}(\mathfrak{K})$. By Lemma \ref{lemm:lin_indip_g}, these elements are dense in $\mathfrak{A}$, so we conclude that $\mathfrak{K}$ is regular.
	\end{proof}

	\begin{lemma}\label{lema:MASA}
		$\mathfrak{K}$ is MASA if and only if $\mathfrak{K}=\mathfrak{K}'$.
	\end{lemma}
	
	\begin{proof} First,  if $\mathfrak{K}_1$ is an abelian $C^*$-algebra such that $\mathfrak{K}\subseteq \mathfrak{K}_1\subseteq \mathfrak{K}'$, then from $\mathfrak{K}=\mathfrak{K'}$ one has that also $\mathfrak{K}=\mathfrak{K}_1$, namely  $\mathfrak{K}$ is MASA.
		On the other hand, let $\mathfrak{K}$ be MASA and assume that $\mathfrak{K}\not=\mathfrak{K}'$. Then in view of Lemma \ref{Lem:span-K'}
		there is (at least) one element of the form $P_tQ_\gamma$ with $(t,\gamma)\in {\rm Ker}(\delta^0)\times {\rm Ker}(\delta_0)$ that is not in $\mathfrak{K}$. 
		Consider the $C^*$-algebra $\mathfrak{K}_1$ generated by $\mathfrak{K}$ and $P_tQ_\gamma$. This is an abelian $C^*$-algebra because $P_tQ_\gamma$ commute with all the elements of the form $P_{\delta^{-1}l}Q_{\delta_1\chi}$ and these elements generate $\mathfrak{K}$ (see Remark \ref{rk:generators_K}). But this is a contradiction for the maximality of $\mathfrak{K}$. Therefore, $\mathfrak{K}=\mathfrak{K'}$.
	\end{proof}
	
	\medskip
	
	The following characterization is a structural result.

	\begin{prop}\label{porp:struc}
		The following statements are equivalent:
		\begin{enumerate}
			\item[(i)] $\mathfrak{K}$ is MASA;
			\item[(ii)] $\mathfrak{K}=\mathfrak{K}'$;
			\item[(iii)] $\mathfrak{K}$ has the UEP.
		\end{enumerate}
	\end{prop}

	\begin{proof}
		We have already proved in Lemma \ref{lema:MASA} the  equivalence (i) $\Leftrightarrow$ (ii). Moreover, from Lemma \ref{lemm:Sep} it follows that (iii) $\Rightarrow$ (i). Therefore, to conclude the proof it is enough to show that (ii) $\Rightarrow$ (iii). For that let us assume  $\mathfrak{K}=\mathfrak{K}'$. 
		Let $\varsigma$ be 	
		a pure state on $\mathfrak{K}=\mathfrak{K}'$.
		By Lemma \ref{Lem:span-K'}, this state is completely determined by the values on elements of the form $P_tQ_\gamma$ where $(t,\gamma)\in {\rm Ker}(\delta^0)\times {\rm Ker}(\delta_0)$. And for $A_sB_\chi$, 
		\begin{equation}\label{eq:mos_purst}
			1\;=\;\varsigma(\n{1})\;=\;\varsigma(A_tB_\chi(A_tB_\chi)^*)\;=\;\varsigma(A_tB_\chi)\overline{\varsigma(A_tB_\chi)}
		\end{equation}
		where we use that $\varsigma$ is a pure state of an abelian algebra, \emph{i.e.} a character. Then we conclude that $|\varsigma(A_sB_\chi)|=1$. Consider $\Tilde{\varsigma}$ an extension of $\varsigma$ to the $C^*$-algebra $\mathfrak{A}$. Using Lemma \ref{lema estado},  we obtain
		\begin{equation}
			\begin{aligned}
				\Tilde{\varsigma}(B_\chi)\Tilde{\varsigma}(P_tQ_\gamma)\;&=\;\Tilde{\varsigma}(B_\chi P_tQ_\gamma)\\
				&=\;\chi(t)\Tilde{\varsigma}(P_tQ_\gamma)\Tilde{\varsigma}(B_\chi)\;=\;\kappa(\delta^0 t)\Tilde{\varsigma}(P_tQ_\gamma)\Tilde{\varsigma}(B_\chi)
			\end{aligned}
		\end{equation}
		where $\chi\in {\rm Im}({\delta_1})$, \emph{i.e.} $\chi=\delta_1\kappa$ with $\kappa\in {\rm Hom}_0(\f{C},\n{G})_1$. Therefore, if $t\not\in {\rm Ker}(\delta^0)$, $\Tilde{\varsigma}(P_tQ_\gamma)=0$. An analogous argument involving the operators $A_s$ shows that if $\gamma\not\in {\rm Ker}(\delta_0)$, then $\Tilde{\varsigma}(P_tQ_\gamma)=0$.
		By Lemma \ref{lemm:lin_indip_g}, the operators $P_tQ_\gamma$ are linearly independent generators of the dense subalgebra $\mathfrak{A}_{\rm loc}$, so that, by continuity, the value of $\Tilde{\varsigma}$ is completely determined by its values over elements of the form $P_tQ_\gamma$. Hence, it is enough to know the values $\Tilde{\varsigma}(P_tQ_\gamma)$ when $(t,\gamma)\in {\rm Ker}(\delta^0)\times {\rm Ker}(\delta_0)$. But, by our hypothesis, $P_tQ_\gamma\in\mathfrak{K}$ if $(t,\gamma)\in {\rm Ker}(\delta^0)\times {\rm Ker}(\delta_0)$. Therefore, any extension $\Tilde{\varsigma}$ of $\varsigma$ is completely determined by
		\begin{equation}\label{eq:unique_extension}
			\Tilde{\varsigma}(P_tQ_\gamma)\;=\;
			\left\{
			\begin{aligned}
				&\varsigma(P_tQ_\gamma)&&{\rm if}\;\;(t,\gamma)\in {\rm Ker}(\delta^0)\times {\rm Ker}(\delta_0)\;,\\
				&0&&{\rm if}\;\;(t,\gamma)\not\in {\rm Ker}(\delta^0)\times {\rm Ker}(\delta_0)\;.
			\end{aligned}
			\right.
		\end{equation}
		In other words, there is only one extension for $\varsigma$ to the whole algebra, and this extension is given by Eq.\eqref{eq:unique_extension}. In particular, this extension is \emph{the} pure one, since every pure state admits at least one pure extension \cite[Proposition 2.3.24]{Bratelli}. Therefore one concludes that $\mathfrak{K}$ has the UEP.
	\end{proof}

	\begin{remark}[Extension of the  seed state]
		By Corollary \ref{cor_uni_ext}, we know that 
		\[
		\mathtt{FG}_{\rm QDM}(\mathfrak{A})\;\simeq\;\mathtt{S}_{\vartheta}(\mathfrak{K}')\;=\;\mathtt{S}_{\vartheta}(\mathfrak{K})\;=\;\{\vartheta\}
		\] Then, there is a unique state in $\mathtt{FG}_{\rm QDM}(\mathfrak{A})$ and by Proposition \ref{prop:Exixst}, any extension of $\vartheta$ to the $C^*$-algebra $\mathfrak{A}$ will be the unique frustration-free ground state, \emph{i.e.} there is a unique extension. In this sense Proposition \ref{porp:struc} provides the unique extension of the seed state $\vartheta$. In particular for the seed state $\vartheta$ the equation \eqref{eq:mos_purst} is satisfied with $\vartheta(A_sB_\chi)=1$.\hfill $\blacktriangleleft$
	\end{remark}

	\medskip
	
	Combining 
	all the previous results one gets that  any of the conditions in 
	Proposition \ref{porp:struc} is sufficient to guarantee that $\mathfrak{K}$  is a $C^*$-diagonal in $\mathfrak{A}$. This is the crucial observation to show that the $C^*$-diagonality of $\mathfrak{K}$  has a 
	purely (co)homological characterization.

	\begin{teo}[Topological characterization of $C^*$-diagonality]
		The unital $C^*$-algebra $\mathfrak{K}$ is a $C^*$-diagonal in $\mathfrak{A}$ if and only if
		both
		$\mathring{H}^0(\f{C},\n{G})$ and $\mathring{H}_0(\f{C},\n{G})$
		are trivial.
	\end{teo}
	
	\begin{proof}
		First, notice that if $\mathring{H}^0(\f{C},\n{G})=\{0\}=\mathring{H}_0(\f{C},\n{G})$ are trivial groups, then ${\rm Ker}(\delta^0)={\rm Im}(\delta^{-1})$ and ${\rm Ker}(\delta_0)={\rm Im}(\delta_{1})$. Then, by Remark \ref{rk:generators_K} and Lemma \ref{Lem:span-K'}, we conclude that $\mathfrak{K}$ and $\mathfrak{K'}$ have the same generators. Therefore they are equal, and 
		all the properties of Proposition \ref{porp:struc} hold. Since
		Lemma \ref{prop:regular} provides the regularity,
		Lemma \ref{lemm:Sep} applies, and $\mathfrak{K}$ is a $C^*$-diagonal.\\
		On the other hand, suppose that $\mathfrak{K}$ is a $C^*$-diagonal  but the groups are not both trivial. Without loss of generality, let $\mathring{H}^0(\f{C},\n{G})\neq\{0\}$. This means that  ${\rm Ker}(\delta^0)$ is bigger than ${\rm Im}(\delta^{-1})$. Take an element $t\in{\rm Ker}(\delta^0)\backslash{\rm Im}(\delta^{-1})$. Then, $P_{t}\in \mathfrak{K}'=\mathfrak{K}$ and using that the canonical projection $\pi:\mathfrak{K}'\to \mathfrak{A}_{\rm{log}}$ satisfies $\pi(\mathfrak{K})=\n{C}[\n{1}]$ one gets
		\begin{equation}
			\pi(P_t)\;=\;P_{[t]}\;=\;c[\n{1}]\;=\;cP_{[0]}
		\end{equation}
		But this is a contradiction in view of the linear independence of  
		the family of operators $\{P_{[t]}Q_{[\gamma]}\}$ by Lemma \ref{lem:linear_independence_logical}.
	\end{proof}

	\section{Appendix A: Reduction to known models}
	
	\subsection{Quantum Double Models}
	
	\label{sec:quantum_double_models}
	It is worth examining how the quantum double models, originally introduced by Kitaev in \cite{Kitaev_2003} and later cast in the frustration-free Hamiltonian language we also use here by \cite{Ferreira:2015uia}, can be understood as a special case within the context of our model. This will help to clarify the relationship between the setting developed in this work and the celebrated quantum double construction.  To do this, consider Fig. \ref{fig:Quantum Double Models}:

	\begin{figure}[H]
		\centering
		\begin{tikzcd}[row sep=1cm, column sep=2cm]
			0\arrow{r} & \f{C}_2\arrow{r}{\partial_2^\f{C}}\arrow{d}{f_2}& \f{C}_1\arrow{r}{\partial_1^\f{C}}\arrow{d}{f_1}&  \f{C}_0\arrow{r}{\partial_0^\f{C}}\arrow{d}{f_0}&0\\
			0\arrow{r}& \n{G}_2\arrow{r}[swap]{\partial_2^\n{G}}& \n{G}_1\arrow{r}[swap]{\partial_1^\n{G}}&  \n{G}_0\arrow{r}[swap]{\partial_0^\n{G}}&0
		\end{tikzcd}
		\caption{Quantum Double Models}
		\label{fig:Quantum Double Models}
	\end{figure}

	\medskip

	First, let us see how the operator $A_{v^*_g}$ works. Recall that by Definition \ref{def:localized maps}, $v^*_g\in{\rm Hom}_0(\f{C},\n{G})^{-1}$ (see Fig. \ref{fig:v_g^*}) sends the vertex $v\in\f{K}_0$ to the group element $g\in \n{G}_1$, and is zero on 
	all the other elements of the simplicial complex:\\
	\begin{figure}[H]
		%-------------------- inlined from QDM1.tex --------------------
		$$
		\begin{tikzcd}[row sep=1cm, column sep=2cm]
			0\arrow{r} & \f{C}_2\arrow{r}{\partial_2^\f{C}}\arrow{dl}{(v_g^*)_2}& \f{C}_1\arrow{r}{\partial_1^\f{C}}\arrow{dl}{(v_g^*)_1}&  \f{C}_0\arrow{r}{\partial_0^\f{C}}\arrow{dl}{(v_g^*)_0}&0\\
			0\arrow{r}& \n{G}_2\arrow{r}[swap]{\partial_2^\n{G}}& \n{G}_1\arrow{r}[swap]{\partial_1^\n{G}}&  \n{G}_0\arrow{r}[swap]{\partial_0^\n{G}}&0
		\end{tikzcd}
		$$
		%-----------------------------------------------------------------
		\caption{$v_g^*$: the localized map assigning $g\in\n{G}_1$ only to the vertex $v$ and zero to all other elements of the simplicial complex.}
		\label{fig:v_g^*}
	\end{figure}

	\medskip
	
	So 
	\begin{equation*}
		(\delta^{-1}v^*_g)_n\;=\;(v^*_g)_{n-1}\circ \partial_n^\f{C}+\partial_{n+1}^\n{G}(v^*_g)_n
	\end{equation*}
	is trivial for $n\geqslant 2$. For $n=1$, one has
	\begin{equation*}
		(\delta^{-1}v^*_g)_1\;=\;(v^*_g)_{0}\circ \partial_n^\f{C}\end{equation*}
	using the fact that also $(v^*_g)_1$ is trivial.
	The operator $(\delta^{-1}v^*_g)_1$ maps all the edges to zero, except those that have $v$ on their boundary. These edges are mapped to $\pm g$, with the sign depending on their orientation.
	For $n=0$, one has
	\begin{equation*}
		(\delta^{-1}v^*_g)_0\;=\;\partial_{1}^\n{G}(v^*_g)_0
	\end{equation*}
	which maps all the vertices to zero, except $v$. The latter is mapped to $\delta_1^\n{G}(g)$.

	\medskip

	So, for $\ket{f}\in \mathfrak{H}_\Lambda$, $A_{v^*_g}\ket{f}:=\ket{f+\delta^{-1}v^*_g}$ would act trivially except in $v\in \Lambda$ and in the edges contained in $\Lambda$ that have $v$ in their boundary. In that case (assuming $\Lambda$ large enough), $A_{v^*_g}$ acts like in Fig. \ref{fig:A_gv}.

	\begin{figure}[H]
		\centering
		\includegraphics[scale=0.5]{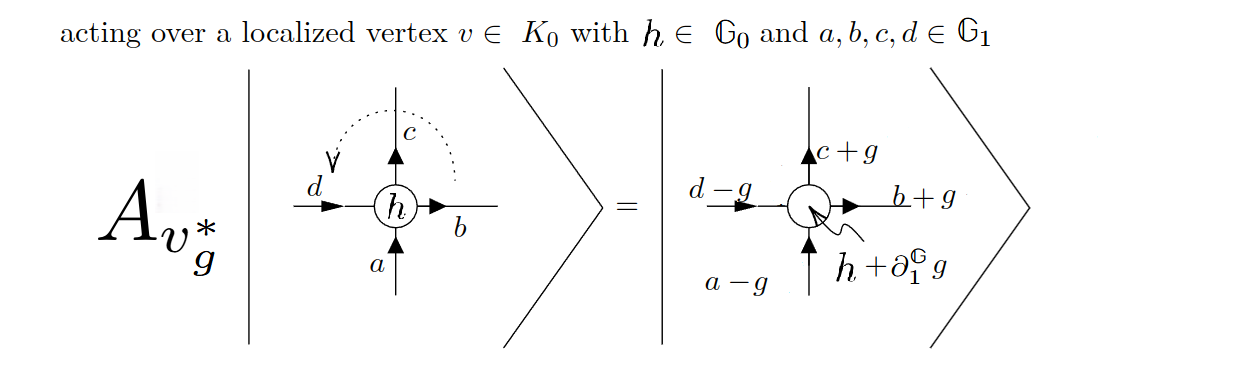}
		\caption{\label{fig:A_gv} Local compactly supported \textit{vertex shift operator}.}
	\end{figure}
	\begin{figure}[H]
		\centering
		\includegraphics[scale=0.5]{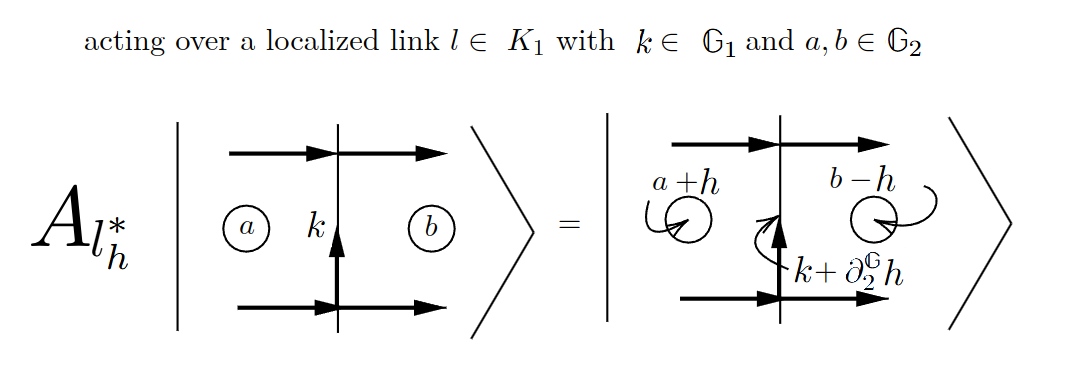}
		\caption{\label{fig:A_hl} Locally supported \textit{ link shift operator}.}
	\end{figure}

	%%%%%%%%%%%%%%%%%%%%%%%%%%%%%%%%%
	%%%%%%%%%%%%%%%AGREGAR DIAGRAMA%%%%%%%%%%%%%%%%
	%%%%%%%%%%%%%%%%%%%%%%%%%%%%%%%%%%%%%%%%%%%%%%%%%%

	Now, let us see how $A_{l_h^*}$ works for the edge $l\in \f{K}_1$ and the group element $h\in \n{G}_2$. In this case, $l_h^*$ sends $l$ to $h$, and all the other elements are sent to zero. So, for $n=2$, one has
	
	\begin{equation*}
		(\delta^{-1}l_h^*)_2\;=\;(l_h^*)_{1}\circ\partial_2^\f{C}\;.
	\end{equation*}
	So the faces are sent to 0, unless they have $l$ in their boundary. If this happens, they are sent to $\pm h$ according to the orientation. For $n=1$, one has 
	\begin{equation*}
		(\delta^{-1}l_h^*)_1\;=\:\partial_{2}^\n{G}(l_h^*)_1
	\end{equation*}
	so $l$ is sent to $\partial_2^\n{G}(h)$. All the other edges are sent to 0. And for $n=0$ and $n\geqslant 3$, the map is trivial.
	Thus, $A_{l_h^*}$ acts as the identity, except in $l$ and its coboundary. Here, the operator acts as shown in Fig. \ref{fig:A_hl}.

	\medskip

	In conclusion, the operators $A_v^0$ and $A_l^0$, introduced in Definition \ref{def:local_operators} look like Fig. \ref{fig:A_v^0 operator} and Fig. \ref{fig:A_l^0 operator}:

	\begin{figure}[H]
		\centering
		\includegraphics[scale=0.45]{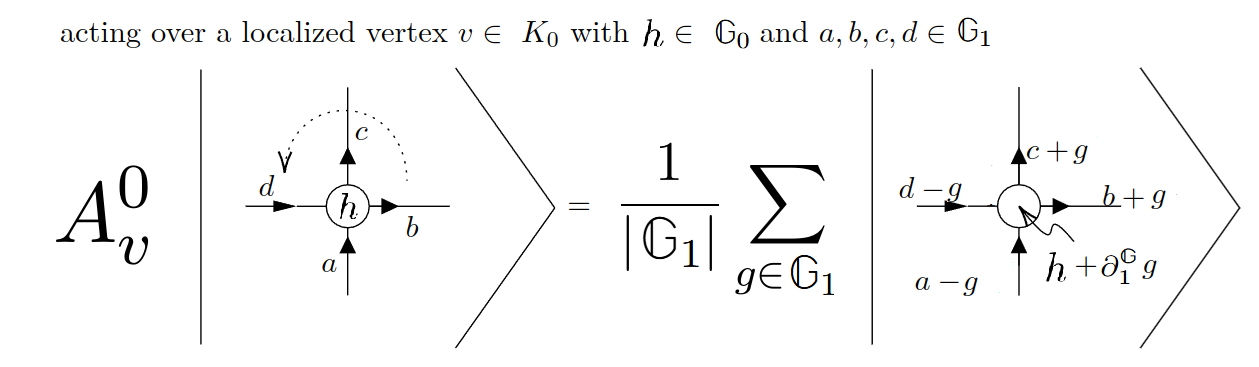}
		\caption{\label{fig:A_v^0 operator} Local compactly supported \textit{vertex star operator}.}
	\end{figure}
	\begin{figure}[H]
		\centering
		\includegraphics[scale=0.45]{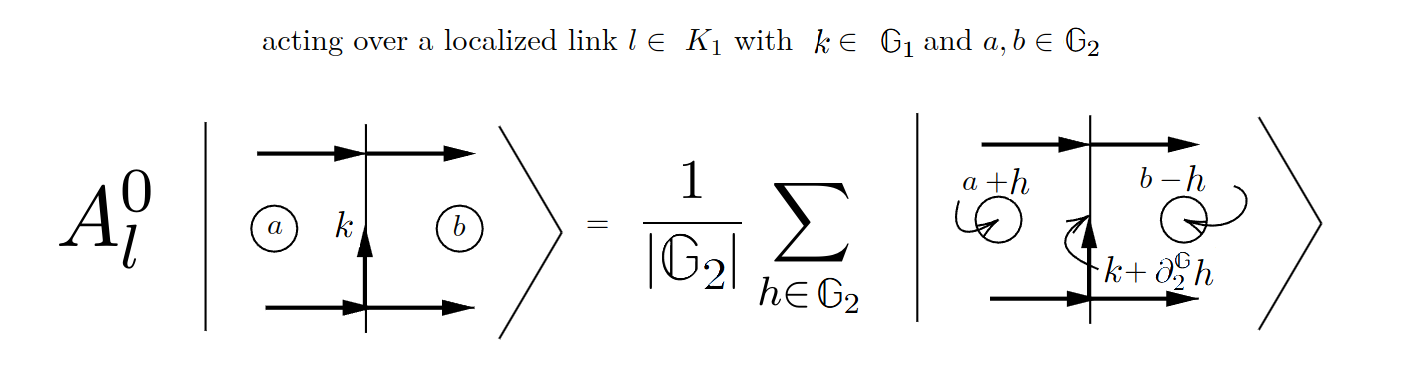}
		\caption{\label{fig:A_l^0 operator} Locally supported \textit{link star operator}.}
	\end{figure}

	Note that in the case $\n{G}_0=\n{G}_2=0$, one recovers the usual star operators of the \emph{toric code model} \cite{Kitaev_2003}.

	%%%%%%%%%%%%%%%%%%%%%%%%%%%%%%
	%%%%%%%%%%%%%%%%%%%%%%%%%%%%%%

	\medskip

	%%%%%%%%%%%%%%

	%------------  UP TO HERE ----------

	%%%%%%%%%%%%%%%%

	\medskip

	The same can be done for the dual assignments in ${\rm Hom}_0(\f{C},\n{G})_p$. According to Definition \ref{def:localized maps}, consider the map $p_*^\gamma\in {\rm Hom}_0(\f{C},\n{G})_1$ (see Fig. \ref{fig:p_{*}}), which sends the face $p\in \f{K}_2$ to the character $\gamma\in\widehat{\n{G}}_1$ and is zero on all the other elements of the simplicial complex:

	\begin{figure}[H]
		\centering
		%-------------------- inlined from QDM2.tex --------------------
		$$
		\begin{tikzcd}[row sep=1cm, column sep=2cm]
			0\arrow{dr}{0}&\arrow{l}  \f{C}_2\arrow{dr}{(p_{*}^\gamma)_2}&\arrow{l}[swap]{\hat{\partial}_2^\f{C}} \f{C}_1\arrow{dr}{(p_{*}^\gamma)_1} &\arrow{l}[swap]{\hat{\partial}_1^\f{C}}\f{C}_0 \arrow{dr}{(p_{*}^\gamma)_0}&\arrow{l}[swap]{\hat{\partial}_0^\f{C}}0\\
			0&\arrow{l} \hat{\n{G}}_2&\arrow{l}{\hat{\partial}_2^{\n{G}}} \hat{\n{G}}_1&\arrow{l}{\hat{\partial}_1^{\n{G}}} \hat{\n{G}}_0&\arrow{l}{\hat{\partial}_0^{\n{G}}} 0
		\end{tikzcd}
		$$
		%-----------------------------------------------------------------
		\caption{\label{fig:p_{*}} $p_{*}^\gamma$: the dual localized map assigning the character $\gamma\in\widehat{\n{G}}_1$ to the face $p$ and the trivial dual element to all other elements of the simplicial complex.}
	\end{figure}

	\noindent
	Let us compute $(\delta_1p_*^\gamma)_n$. For $n=2$, one has
	\begin{equation*}
		(\delta_1p_*^\gamma)_2\;=\;\widehat{\partial}^{\n{G}}_2\circ (p_*^\gamma)_2
	\end{equation*}
	since $(p_*^\gamma)_3=0$.
	So $(\delta_1p_*^\gamma)_2$ sends $p$  to $\widehat{\partial}^{\n{G}}_2(\gamma)$ and the other faces are sent to zero. For $n=1$,
	\begin{equation*}
		(\delta_1(p_*^\gamma))_1\;=\;( p_*^\gamma)_2\circ \widehat{\partial}^\f{C}_2
	\end{equation*}
	since $(p_*^\gamma)_1=0$. Therefore, $(\delta_1( p_*^\gamma))_1$ sends the boundary of $p$  to $\pm \gamma$, where the sign depends on the orientation. Finally $(\delta_1p_*^\gamma)_n=0$ if $n=0$ and $n\geqslant3$.
	So $B_{p_*^\gamma}$ only depends on $p$ and its boundary. This can be seen in Fig. \ref{fig:B_gammap}.
	
	\begin{figure}[H]
		\centering
		\begin{tikzpicture}%[scale=0.8]
			% Dibuja el triángulo equilátero con base 3 y flechas más gruesas
			\draw[->, thick] (0,0) -- (2.8,0) node[midway, below] {$a$}; % Flecha más corta
			\draw[->, thick] (3,0) -- (1.6,2.498) node[midway, right] {$b$}; % Flecha más corta
			\draw[->, thick]   (0.1,0.1) --(1.4,2.498) node[midway, left] {$c$}; % Flecha más corta
			
			% Dibuja los vértices como círculos pequeños de color gris claro
			\filldraw[gray!40] (0,0) circle (3pt);
			\filldraw[gray!40] (3,0) circle (3pt);
			\filldraw[gray!40] (1.5,2.598) circle (3pt);
			
			% Dibuja el círculo en el centro del triángulo con la letra "g"
			\node[draw, circle, inner sep=2pt] at (1.5,1.299) {$g$}; % Círculo con "g"
			% Dibuja las líneas verticales para el ket
			\draw[thick] (-0.5, 0) -- (-0.5, 2.598); % Línea izquierda % Línea izquierda
			
			% Línea derecha (superior) en un ángulo de 150 grados
			\draw[thick] (3.5,0) -- (3.8, 1.299); % Primera línea a la derecha (en ángulo de 150 grados)
			
			% Línea derecha (inferior) en un ángulo de -150 grados
			\draw[thick] (3.8,1.299) -- (3.5, 2.598); % Segunda línea a la derecha (en ángulo de -150 grados)
			%Coloca el texto B_m a la izquierda de la línea vertical, a la altura de la mitad del triángulo
			\node[scale=1.5] at (-1.5, 1.299) { $B_{ p_*^\gamma}$}; % B_m más grande
			\node[scale=1] at (5.5, 1.299) { $\;\;\;=\gamma(hol(p)+\partial_2^{\n{G}} g)$}; % = más grande
			
			% Dibuja el triángulo equilátero con base 3 y flechas más gruesas
			\draw[->, thick] (8,0) -- (10.8,0) node[midway, below] {$a$}; % Flecha más corta
			\draw[->, thick] (11,0) -- (9.6,2.498) node[midway, right] {$b$}; % Flecha más corta
			\draw[->, thick] (8.1,0.1) -- (9.4,2.498) node[midway, left] {$c$}; % Flecha más corta
			
			% Dibuja los vértices como círculos pequeños de color gris claro
			\filldraw[gray!40] (8,0) circle (3pt);
			\filldraw[gray!40] (11,0) circle (3pt);
			\filldraw[gray!40] (9.5,2.598) circle (3pt);
			
			% Dibuja el círculo en el centro del triángulo con la letra "g"
			\node[draw, circle, inner sep=2pt] at (9.5,1.299) {$g$}; % Círculo con "g"
			% Dibuja las líneas verticales para el ket
			\draw[thick] (7.5, 0) -- (7.5, 2.598); % Línea izquierda % Línea izquierda
			
			% Línea derecha (superior) en un ángulo de 150 grados
			\draw[thick] (11.5,0) -- (11.8, 1.299); % Primera línea a la derecha (en ángulo de 150 grados)
			
			% Línea derecha (inferior) en un ángulo de -150 grados
			\draw[thick] (11.8,1.299) -- (11.5, 2.598); % Segunda línea a la derecha (en ángulo de -150 grados)
		\end{tikzpicture}
		\caption{\label{fig:B_gammap} Locally supported \textit{plaquette clock operator}, acting over a localized plaquette $p\in \f{K}_2$ with $g\in\n{G}_2$ and $a,b,c\in\n{G}_1$ as labelled, where $hol(p):= a+b-c$. }
	\end{figure}

	Finally, consider $l^\alpha_*\in {\rm Hom}_0(\f{C},\n{G})_1$
	which sends the edge $l\in \f{K}_1$ to the character $\alpha\in\widehat{\n{G}}_0$ and is zero on 
	all the other elements of the simplicial complex.
	In this case, one has that
	\begin{equation*}
		(\delta_1l^\alpha_*)_1\;=\;\widehat{\partial}^\n{G}_1\circ (l^\alpha_*)_1
	\end{equation*}
	since $(l^\alpha_*)_2=0$.
	So $(\delta_1l^\alpha_*)_1$ sends all the edges to zero, except  $l$. The latter is sent to $\widehat{\partial}^\n{G}_1(\alpha)$. For
	$n=0$
	\begin{equation*}
		(\delta_1l^\alpha_*)_0\;=\;(l^\alpha_*)_1\circ \widehat{\partial}^\f{C}_1
	\end{equation*}
	since $(l^\alpha_*)_0=0$. Then, $(\delta_1l^\alpha_*)_0$
	acts trivially on all the vertices, except the ones in the boundary of $l$. The latter are sent to $\pm \alpha$ according to the orientation.
	Finally $(\delta_1l^\alpha_*)_n=0$ for all $n\geqslant2$.
	This is shown in Fig. \ref{fig:B_alphal}.

	\begin{figure}[H]
		\centering
		\begin{tikzpicture}
			% Dibuja la flecha horizontal
			\draw[->, thick] (0.2,1.299) -- (2.7,1.299)node[midway, below] {$g$}; % Flecha horizontal
			
			\node[draw, circle, inner sep=2pt] at (0,1.299) {$a$}; % Círculo con "a"
			\node[draw, circle, inner sep=2pt] at (3,1.299) {$b$}; % Círculo con "b"
			
			% Dibuja las líneas verticales para el ket
			\draw[thick] (-0.5, 0) -- (-0.5, 2.598); % Línea izquierda % Línea izquierda
			
			% Línea derecha (superior) en un ángulo de 150 grados
			\draw[thick] (3.5,0) -- (3.8, 1.299); % Primera línea a la derecha (en ángulo de 150 grados)
			
			% Línea derecha (inferior) en un ángulo de -150 grados
			\draw[thick] (3.8,1.299) -- (3.5, 2.598); % Segunda línea a la derecha (en ángulo de -150 grados)
			%Coloca el texto B_m a la izquierda de la línea vertical, a la altura de la mitad del triángulo
			\node[scale=1.5] at (-1.5, 1.299) { $B_{ l_*^\alpha}$}; % B_m más grande
			\node[scale=1] at (5.5, 1.299) { $=\alpha(b-a+\partial_1^{\n{G}} g)$}; % = más grande
			
			% Dibuja la flecha horizontal
			\draw[->, thick] (8.2,1.299) -- (10.7,1.299)node[midway, below] {$g$}; % Flecha horizontal
			
			\node[draw, circle, inner sep=2pt] at (8,1.299) {$a$}; % Círculo con "a"
			\node[draw, circle, inner sep=2pt] at (11,1.299) {$b$}; % Círculo con "b"
			% Dibuja las líneas verticales para el ket
			\draw[thick] (7.5, 0) -- (7.5, 2.598); % Línea izquierda % Línea izquierda
			
			% Línea derecha (superior) en un ángulo de 150 grados
			\draw[thick] (11.5,0) -- (11.8, 1.299); % Primera línea a la derecha (en ángulo de 150 grados)
			
			% Línea derecha (inferior) en un ángulo de -150 grados
			\draw[thick] (11.8,1.299) -- (11.5, 2.598); % Segunda línea a la derecha (en ángulo de -150 grados)
		\end{tikzpicture}
		\caption{\label{fig:B_alphal} Locally supported \textit{link clock operator}, acting over a localized link $l\in \f{K}_1$ with $g\in\n{G}_1$ and $a,b\in\n{G}_0$ as labelled.}
	\end{figure}

	Therefore, using the orthogonal relation of the characters (\cite{issacs}), the operators $B_p^0$ and $B_l^0$ outlined in Definition \ref{def:local_operators} look like Fig. \ref{fig:B_p^0} and \ref{fig:B_l^0}:
	\begin{figure}[H]
		\centering
		\begin{tikzpicture}
			% Dibuja el triángulo equilátero con base 3 y flechas más gruesas
			\draw[->, thick] (0,0) -- (2.8,0) node[midway, below] {$a$}; % Flecha más corta
			\draw[->, thick] (3,0) -- (1.6,2.498) node[midway, right] {$b$}; % Flecha más corta
			\draw[->, thick]  (0.1,0.1)--(1.4,2.498) node[midway, left] {$c$}; % Flecha más corta
			
			% Dibuja los vértices como círculos pequeños de color gris claro
			\filldraw[gray!40] (0,0) circle (3pt);
			\filldraw[gray!40] (3,0) circle (3pt);
			\filldraw[gray!40] (1.5,2.598) circle (3pt);
			
			% Dibuja el círculo en el centro del triángulo con la letra "g"
			\node[draw, circle, inner sep=2pt] at (1.5,1.299) {$g$}; % Círculo con "g"
			% Dibuja las líneas verticales para el ket
			\draw[thick] (-0.5, 0) -- (-0.5, 2.598); % Línea izquierda % Línea izquierda
			
			% Línea derecha (superior) en un ángulo de 150 grados
			\draw[thick] (3.5,0) -- (3.8, 1.299); % Primera línea a la derecha (en ángulo de 150 grados)
			
			% Línea derecha (inferior) en un ángulo de -150 grados
			\draw[thick] (3.8,1.299) -- (3.5, 2.598); % Segunda línea a la derecha (en ángulo de -150 grados)
			%Coloca el texto B_m a la izquierda de la línea vertical, a la altura de la mitad del triángulo
			\node[scale=1.5] at (-1.5, 1.299) { $B_{p }^0$}; % B_m más grande
			\node[scale=1] at (5.5, 1.299) { $\;\;\;=\delta(-\partial_2^{\n{G}}g,hol(p))$}; % = más grande
			
			% Dibuja el triángulo equilátero con base 3 y flechas más gruesas
			\draw[->, thick] (8,0) -- (10.8,0) node[midway, below] {$a$}; % Flecha más corta
			\draw[->, thick] (11,0) -- (9.6,2.498) node[midway, right] {$b$}; % Flecha más corta
			\draw[->, thick] (8.1,0.1)--(9.4,2.498) node[midway, left] {$c$}; % Flecha más corta
			
			% Dibuja los vértices como círculos pequeños de color gris claro
			\filldraw[gray!40] (8,0) circle (3pt);
			\filldraw[gray!40] (11,0) circle (3pt);
			\filldraw[gray!40] (9.5,2.598) circle (3pt);
			
			% Dibuja el círculo en el centro del triángulo con la letra "g"
			\node[draw, circle, inner sep=2pt] at (9.5,1.299) {$g$}; % Círculo con "g"
			% Dibuja las líneas verticales para el ket
			\draw[thick] (7.5, 0) -- (7.5, 2.598); % Línea izquierda % Línea izquierda
			
			% Línea derecha (superior) en un ángulo de 150 grados
			\draw[thick] (11.5,0) -- (11.8, 1.299); % Primera línea a la derecha (en ángulo de 150 grados)
			
			% Línea derecha (inferior) en un ángulo de -150 grados
			\draw[thick] (11.8,1.299) -- (11.5, 2.598); % Segunda línea a la derecha (en ángulo de -150 grados)
		\end{tikzpicture}
		\caption{\label{fig:B_p^0} Locally supported \textit{plaquette operator}, acting over a localized plaquette $p\in \f{K}_2$ with $g\in\n{G}_2$ and $a,b,c\in\n{G}_1$ as labelled. Again $hol(p):=a+b-c$. }
	\end{figure}
	\begin{figure}[H]
		\centering
		\begin{tikzpicture}
			% Dibuja la flecha horizontal
			\draw[->, thick] (0.2,1.299) -- (2.7,1.299)node[midway, below] {$g$}; % Flecha horizontal
			
			\node[draw, circle, inner sep=2pt] at (0,1.299) {$a$}; % Círculo con "a"
			\node[draw, circle, inner sep=2pt] at (3,1.299) {$b$}; % Círculo con "b"
			
			% Dibuja las líneas verticales para el ket
			\draw[thick] (-0.5, 0) -- (-0.5, 2.598); % Línea izquierda % Línea izquierda
			
			% Línea derecha (superior) en un ángulo de 150 grados
			\draw[thick] (3.5,0) -- (3.8, 1.299); % Primera línea a la derecha (en ángulo de 150 grados)
			
			% Línea derecha (inferior) en un ángulo de -150 grados
			\draw[thick] (3.8,1.299) -- (3.5, 2.598); % Segunda línea a la derecha (en ángulo de -150 grados)
			%Coloca el texto B_m a la izquierda de la línea vertical, a la altura de la mitad del triángulo
			\node[scale=1.5] at (-1.5, 1.299) { $B_{ l}^0$}; % B_m más grande
			\node[scale=1] at (5.5, 1.299) { $=\delta(-\partial_1^{\n{G}} g, b-a)$}; % = más grande
			
			% Dibuja la flecha horizontal
			\draw[->, thick] (8.2,1.299) -- (10.7,1.299)node[midway, below] {$g$}; % Flecha horizontal
			
			\node[draw, circle, inner sep=2pt] at (8,1.299) {$a$}; % Círculo con "a"
			\node[draw, circle, inner sep=2pt] at (11,1.299) {$b$}; % Círculo con "b"
			% Dibuja las líneas verticales para el ket
			\draw[thick] (7.5, 0) -- (7.5, 2.598); % Línea izquierda % Línea izquierda
			
			% Línea derecha (superior) en un ángulo de 150 grados
			\draw[thick] (11.5,0) -- (11.8, 1.299); % Primera línea a la derecha (en ángulo de 150 grados)
			
			% Línea derecha (inferior) en un ángulo de -150 grados
			\draw[thick] (11.8,1.299) -- (11.5, 2.598); % Segunda línea a la derecha (en ángulo de -150 grados)
		\end{tikzpicture}
		\caption{\label{fig:B_l^0} Locally supported \textit{link operator}, acting over a localized link $l\in \f{K}_1$ with $g\in\n{G}_1$ and $a,b\in\n{G}_0$ as labelled.}
	\end{figure}

	\subsection{Appendix B: Higher dimensional models}\label{sec:models}
	
	Our model also recovers other interesting models existing in the literature. For example the $3D$ toric code \cite{Resende} and the $4D$ loop toric code \cite{Duivenvoorden}. For the first one, we take the simplicial complex as the lattice in $\ZZ^3$ and we consider the following chain complex,

	\begin{figure}[H]
		\centering
		\begin{tikzcd}[row sep=1cm, column sep=2cm]
			0\arrow{r} &\f{C}_3\arrow{r}{\partial_3^\f{C}}\arrow{d}{f_3}& \f{C}_2\arrow{r}{\partial_2^\f{C}}\arrow{d}{f_2}& \f{C}_1\arrow{r}{\partial_1^\f{C}}\arrow{d}{f_1}&  \f{C}_0\arrow{r}{\partial_0^\f{C}}\arrow{d}{f_0}&0\\
			0\arrow{r}&0\arrow{r}[swap]{\partial_3^\n{G}}& 0\arrow{r}[swap]{\partial_2^\n{G}}& \n{G}_1\arrow{r}[swap]{\partial_1^\n{G}}&  0\arrow{r}[swap]{\partial_0^\n{G}}&0
		\end{tikzcd}
		\caption{$3D$ toric code. Here the star operator centered at a vertex $v$ acts non-trivially on the 6 edges that are the coboundary of that vertex. And the plaquette operator for a face $p$ acts non-trivially on the 4 edges that are in the boundary of $p$.}
		\label{fig:3D_toric_code}
	\end{figure}

	For the $4D$ loop toric code, we take the simplicial complex as a lattice in $\ZZ^4$ and we consider,
	
	\begin{figure}[H]
		\centering
		\begin{tikzcd}[row sep=0.5cm, column sep=1.5cm]
			0\arrow{r} &\f{C}_4\arrow{r}{\partial_4^\f{C}}\arrow{d}{f_4}&\f{C}_3\arrow{r}{\partial_3^\f{C}}\arrow{d}{f_3}& \f{C}_2\arrow{r}{\partial_2^\f{C}}\arrow{d}{f_2}& \f{C}_1\arrow{r}{\partial_1^\f{C}}\arrow{d}{f_1}&  \f{C}_0\arrow{r}{\partial_0^\f{C}}\arrow{d}{f_0}&0\\
			0\arrow{r}&0\arrow{r}[swap]{\partial_4^\n{G}}& 0\arrow{r}[swap]{\partial_3^\n{G}}& \n{G}_2\arrow{r}[swap]{\partial_2^\n{G}}& 0\arrow{r}[swap]{\partial_1^\n{G}}&  0\arrow{r}[swap]{\partial_0^\n{G}}&0
		\end{tikzcd}
		\caption{$4D$ toric code. Here the star operator centered at an edge $e$ acts non-trivially on the 6 faces that are the coboundary of that edge. And the plaquette operator for a cube $c$ acts non-trivially on the 6 faces that are in the boundary of $c$.}
		\label{fig:4D_toric_code}
	\end{figure}

	\section{Appendix C: Simplicial Complex}\label{sec:simplicial_complex}

	The concept of a \emph{simplicial complex} serves to generalize the notion of triangulation, providing a powerful framework for understanding and analyzing topological spaces. In this discussion, we shall present only a brief and introductory overview of this topic. For a more comprehensive and detailed treatment, we refer the reader to the classical monographs \cite{Hatcher,Rotman} and, in particular, to \cite[Section 3.1]{spanier}.

	\medskip

	Let ${V}$ be a set. A  {simplicial complex} $\f{K}$ is a family of finite nonempty subsets of ${V}$, called \emph{simplices}, such that
	\begin{enumerate}
		\item if $v\in V$, then $\{v\}\in \f{K}$;
		\item if $s\in \f{K}$ and $s'\subset s$ with $s'\neq\emptyset$, then $s'\in \f{K}$.
	\end{enumerate}
	One refers to  $V$ as the \emph{vertex set} of $\f{K}$ and sometimes the notation $V= {\rm Vert}(\f{K})$ is used. A simplex $s\in \f{K}$ containing exactly $n+1$ distinct vertices is called an $n$-simplex. The subset of all $n$-simplices is denoted with $\f{K}_n\subset K$.
	If $s'\subset s$ then $s'$ is called a \emph{face} of $s$.
	
	\medskip

	We are in a position to introduce the concept of 	\emph{oriented} simplicial complex
	as in 	\cite[Section 4.1]{spanier}. An oriented $n$-simplex of $\f{K}$ is an $n$-simplex $s\in \f{K}_n$ together with an equivalence class of total ordering of the vertices of $s$, two orderings being equivalent if they differ by an even permutation of the vertices.
	If $v_0,v_1\ldots,v_n$	are the vertices of $s$, then $[v_0,\dots,v_n]$ denotes the oriented 
	$n$-simplex $s$ together with the equivalence class of the ordering 
	$v_0<v_1<\ldots<v_n$. For every vertex $v\in V$ there is a unique oriented 0-simplex $[v]$. For every $n$-simplex with $n\geqslant 1$ there correspond exactly two oriented 
	$n$-simplices. If  $[v_0,\dots,v_n]$ denotes an oriented $n$-simplex, then $-[v_0,\dots,v_n]$ will denote the same 	simplex with the opposite orientation obtained by an odd permutation  of the ordering 
	$v_0<v_1<\ldots<v_n$. Finally, we will use   $[v_0,\dots,\hat{v}_i,\dots,v_n]$ to denote
	the oriented $(n-1)$-simplex obtained by removing the vertex $v_i$ from $[v_0,\dots,v_n]$.  From now on, with a little abuse of notation, we will use the same symbol $\f{K}_n$ to indicate the set of oriented $n$-simplex. 
	In Fig. \ref{fig:simplices}, one can see the geometric realization of the building blocks for {oriented} simplicial complexes.

	\begin{figure}[H]
		\centering
		\begin{tikzpicture}
			% 0-simplex
			\fill (-2,0) circle (1.5pt) {};
			\node (0) at (-2,-0.5) {$0$-simplex};
			
			% 1-simplex
			\fill (0,0) circle (1.5pt) {};
			\fill (0,1) circle (1.5pt) {};
			\draw (0,0) -- (0,1) {};
			\draw[->] (0,0) -- (0,0.5) {}; 
			\node (1) at (0,-0.5) {$1$-simplex};
			
			% 2-simplex
			\fill (2,0) circle (1.5pt) {};
			\fill (4,0) circle (1.5pt) {};
			\fill (3,1.732) circle (1.5pt) {}; % Altura ajustada para simetría
			\path[shade,draw] (2,0) -- (4,0) -- (3,1.732) -- cycle; % Triángulo equilátero
			\draw[thick, ->] (3.4,0.6) arc (0:250:0.3);
			\node (2) at (3,-0.5) {$2$-simplex};
		\end{tikzpicture}
		\caption{\label{fig:simplices}Basic oriented simplices for the analysis of (two-dimensional) surfaces.}
	\end{figure}

	\medskip

	Denote by $\f{C}_n$ the free abelian group with basis the oriented $n$-simplices of $\f{K}_n$
	with the relations $s_1+s_2=0$ if $s_1$ and $s_2$ are   oriented $n$-complexes determined by the same vertices but with opposite orientation. The \emph{boundary operator} $\partial_n^\f{C} : \f{C}_n \to \f{C}_{n-1}$ captures how an oriented $n$-simplex interacts with its oriented lower-dimensional faces. Specifically, it maps an $n$-simplex to a signed sum of its $(n{-}1)$-dimensional faces, respecting orientation and structure.
	More precisely, let $x=[v_0,\dots,v_n]\in \f{K}_n$. Then, by definition
	\begin{equation}\label{def:boundary}
		\partial_n^\f{C} x \;:=\; \sum_{y \in \f{K}_{n-1}} \epsilon_{x,y}\;  y
	\end{equation}
	where
	\[
	\epsilon_{x,y}\;: =\;
	\begin{cases}
		+1&\;  \text{if } y =[v_0, \dots, \hat{v}_i, \dots, v_n]\; \text{ for some even } i\;, \\
		-1&\;  \text{if } y =  [v_0, \dots, \hat{v}_i, \dots, v_n]\; \text{ for some odd } i\;, \\
		\phantom{+}0&\;  \text{if }\; y \text{ is not a face of } x\;.
	\end{cases}
	\]
	A direct computation shows that $\partial_n^\f{C}\circ   \partial_{n+1}^\f{C}=0$ for every $n\in\n{N}_0$.

	\medskip

	We also need the coboundary operator $\widehat{\partial}^\f{C}_{n+1}:\f{C}_n\to \f{C}_{n+1}$ defined by 
	\begin{equation}\label{def:coboundary}
		\widehat{\partial}_{n+1}^\f{C} x\;:=\;\sum_{z\in \f{K}_{n+1}}\epsilon_{z,x}\;z\;.
	\end{equation}
	Also in this case one gets that $\widehat{\partial}_{n+1}^\f{C}\circ   \widehat{\partial}_{n}^\f{C}=0$ for every $n\in\n{N}_0$.

	\medskip

	We are now in a position to prove the duality between the maps $\delta^p$ and $\delta_{p+1}$
	expressed by Lemma \ref{lem:duality}.
	\begin{proof}[{Proof of Lemma \ref{lem:duality}}]
		The proof is an explicit computation.
		Notice that
		\begin{align*}
			\gamma(\delta^p (f))\;&=\;\prod_{n\in\NN_0}\prod_{x\in \f{K}_n}\gamma_n(x)[(\delta^p f)_n(x)]\\
			&=\;\prod_{n\in\NN_0}\prod_{x\in \f{K}_n}\gamma_n(x) [f_{n-1}(\partial^\f{C}_n (x))-(-1)^p\partial_{n-p}^\n{G} (f_n(x))]\\
			&=\;\frac{\prod_{n\in\NN_0}\prod_{x\in \f{K}_n}\gamma_n(x)[f_{n-1}(\partial_n^\f{C} (x))]}{\prod_{n\in\NN_0}\prod_{x\in \f{K}_n}\gamma_n(x)[(-1)^p\partial_{n-p}^\n{G} (f_n(x))]}\\
			&=\;\frac{\prod_{n\in\NN_0}\prod_{x\in \f{K}_n}\gamma_n(x)[f_{n-1}(\partial_n^\f{C} (x))]}{\prod_{n\in\NN_0}\prod_{x\in \f{K}_n}[\widehat{\partial}_{n-p}^\n{G}(\gamma_n(x))][(-1)^p f_n(x)]}\;.
		\end{align*}    
		On the other hand,
		\begin{align*}
			(\delta_{p+1}\gamma)(f)\;&=\;\prod_{n\in\NN_0}\prod_{x\in \f{K}_n} (\delta_{p+1}\gamma)_n(x)
			[f_n(x)]\\
			&=\;\prod_{n\in\NN_0}\prod_{x\in \f{K}_n}\left(\gamma_{n+1}(\widehat{\partial}_{n+1}^\f{C} (x))-(-1)^{p}\widehat{\partial}_{n-p}^\n{G}( \gamma_n(x))\right)[f_n(x)]\\
			&=\;\frac{\prod_{n\in\NN_0}\prod_{x\in \f{K}_n}\gamma_{n+1}(\widehat{\partial}_{n+1}^\f{C} (x))[f_{n}(x)]}{\prod_{n\in\NN_0}\prod_{x\in \f{K}_n}[\widehat{\partial}_{n-p}^\n{G}(\gamma_n(x))][(-1)^p f_n(x)]}\;.
		\end{align*}
		As the denominators are equal, it is only necessary to show the equality of the numerators.
		Notice that, as $f_{-1}$ is trivial, one deduces that
		\begin{align*}
			\prod_{n\in\NN_0}\prod_{x\in \f{K}_n}\gamma_n(x)[f_{n-1}(\partial_n^\f{C} (x))]\;&=\;  \prod_{n\in\NN}\prod_{x\in \f{K}_n}\gamma_n(x)[f_{n-1}(\partial_n^\f{C} (x))]\\\
			&=\;\prod_{n\in\NN}\prod_{x\in \f{K}_n}\gamma_n(x)\left[f_{n-1}\left(\sum_{y\in \f{K}_{n-1}}\epsilon_{x,y}\;y\right)\right]\\
			&=\prod_{n\in\NN}\prod_{x\in \f{K}_n}\prod_{y\in \f{K}_{n-1}}\gamma_n(x)[f_{n-1}(\epsilon_{x,y}\;y)]\;.
		\end{align*}
		On the other hand 
		\begin{align*}
			\prod_{n\in\NN_0}\prod_{x\in \f{K}_n}\gamma_{n+1}(\widehat{\partial}_{n+1}^\f{C} (x))[f_{n}(x)]\;&=\;\prod_{n\in\NN_0}\prod_{x\in \f{K}_n}\gamma_{n+1}\left(\sum_{z\in \f{K}_{n+1}} \epsilon_{z,x}\; z\right)[f_{n}(x)]\\
			&=\;\prod_{n\in\NN_0}\prod_{x\in \f{K}_n}\prod_{z\in \f{K}_{n+1}}\gamma_{n+1}( \epsilon_{z,x}\; z)[f_{n}(x)]\\
			&=\;\prod_{n\in\NN_0}\prod_{z\in \f{K}_{n+1}}\prod_{x\in \f{K}_n}\gamma_{n+1}( z)[f_{n}(\epsilon_{z,x}\;x)]\;.
		\end{align*}

		By a change of variable $n\to n-1$ in the second equality one gets the equality of the numerators.
	\end{proof}

	\medskip

	It is important to note that if $\partial^\n{G}_j$ is trivial for all $j \in \ZZ$, then the denominators of both fractions become equal to 1. Consequently, one recovers the usual boundary and coboundary maps used in homological codes (see   \cite[eq. (4.13)]{Vrana} and the subsequent paragraph).


\begin{thebibliography}{99}
		
		\bibitem{2008PhRvL.101a0504L}
		Li, H., Haldane, F.\,D.\,M.:
		{\sl Entanglement spectrum as a generalization of entanglement entropy: Identification of topological order in non-abelian fractional quantum Hall effect states}.
		Phys. Rev. Lett. {\bf 101}, no. 1, 010504 (2008); \eprint{0805.0332}
		
		\bibitem{2010PhRvL.104m0502F}
		Fidkowski, L.:
		{\sl Entanglement spectrum of topological insulators and superconductors}.
		Phys. Rev. Lett. {\bf 104}, no. 13, 130502 (2010); \eprint{0909.2654}
		
		\bibitem{Kitaev_2003}
		Kitaev, A.\,Yu.:
		{\sl Fault-tolerant quantum computation by anyons}.
		Ann. Phys. {\bf 303}, no. 1, 2--30 (2003)
		
		\bibitem{Naaijkens-2011}
		Naaijkens, P.:
		{\sl Localized endomorphisms in Kitaev's toric code on the plane}.
		Rev. Math. Phys. {\bf 23}, no. 4, 347--373 (2011)
		
		\bibitem{Cui_2020}
		Cui, S.\,X., Ding, D., Han, X., Penington, G., Ranard, D., Rayhaun, B.\,C., Zhou, S.:
		{\sl Kitaev's quantum double model as an error correcting code}.
		Quantum {\bf 4} (2020)
		
		\bibitem{Bols_2025}
		Bols, A., Vadnerkar, S.:
		{\sl Classification of the anyon sectors of Kitaev's quantum double model}.
		Commun. Math. Phys. {\bf 406}, no. 8 (2025)
		
		\bibitem{dealmeida2017topological}
		Costa de Almeida, R., Ibieta-Jimenez, J.\,P., Lorca Espiro, J., Teotonio-Sobrinho, P.:
		{\sl Topological order from a cohomological and higher gauge theory perspective}.
		\eprint{1711.04186} (2017)
		
		\bibitem{Vrana}
		Vrana, P., Farkas, M.:
		{\sl Homological codes and abelian anyons}.
		Rev. Math. Phys. {\bf 31}, no. 10, 1950038 (2019)
		
		\bibitem{Lorca2025}
		Lorca Espiro, J.:
		{\sl A dualization approach to the ground state subspace classification of abelian higher gauge symmetry models}.
		J. Math. Phys. {\bf 66} (2025)
		
		
		\bibitem{Bachmann2017}
		Bachmann, S.:
		{\sl Local disorder, topological ground state degeneracy and entanglement entropy, and discrete anyons}.
		Rev. Math. Phys. {\bf 29}, no. 6, 1750018 (2017)
		
		\bibitem{Haah2016}
		Haah, J.:
		{\sl An invariant of topologically ordered states under local unitary transformations}.
		Commun. Math. Phys. {\bf 342}, no. 3, 771--801 (2016)
		
		\bibitem{KatoNaaijkens2020}
		Kato, K., Naaijkens, P.:
		{\sl An entropic invariant for 2D gapped quantum phases}.
		J. Phys. A: Math. Theor. {\bf 53}, no. 8, 085302 (2020)
		
		\bibitem{PoloOjitoProdan2026}
		Polo Ojito, D., Prodan, E.:
		{\sl The full set of KMS-states for abelian Kitaev models}.
		\eprint{2603.27436} (2026)
		
		\bibitem{Pachos}
		Pachos, J.\,K.:
		{\em Introduction to Topological Quantum Computation}.
		Cambridge University Press, Cambridge (2012)
		
		\bibitem{brown}
		Brown, R.:
		{\sl Cohomology with chains as coefficients}.
		Proc. London Math. Soc. (3) {\bf 14}, 545--565 (1964)
		
		\bibitem{Hatcher}
		Hatcher, A.:
		{\em Algebraic Topology}.
		Cambridge University Press, Cambridge (2002)
		
		\bibitem{issacs}
		Isaacs, I.\,M.:
		{\em Character Theory of Finite Groups}.
		Pure and Applied Mathematics, vol.~69, Academic Press, New York-London (1976)
		
		\bibitem{Bratelli}
		Bratteli, O., Robinson, D.\,W.:
		{\em Operator Algebras and Quantum Statistical Mechanics 1: $C^*$- and $W^*$-Algebras, Symmetry Groups, Decomposition of States}.
		Springer-Verlag, New York, 2nd edition (1987)
		
		\bibitem{weil_operators}
		Prasad, A.:
		{\sl On character values and decomposition of the Weil representation associated to a finite abelian group}.
		J. Anal. {\bf 17} (2009)
		
		\bibitem{bratelli2}
		Bratteli, O., Robinson, D.\,W.:
		{\em Operator Algebras and Quantum Statistical Mechanics 2: Equilibrium States. Models in Quantum Statistical Mechanics}.
		Springer-Verlag, Berlin, 2nd edition (1997)
		
		\bibitem{Fujii2015}
		Fujii, K.:
		{\em Quantum Computation with Topological Codes: From Qubit to Topological Fault-Tolerance}.
		SpringerBriefs in Mathematical Physics, vol.~8, Springer, Singapore (2015)
		
		\bibitem{Fannes}
		Alicki, R., Fannes, M., Horodecki, M.:
		{\sl A statistical mechanics view on Kitaev's proposal for quantum memories}.
		J. Phys. A: Math. Theor. {\bf 40}, no. 24, 6451--6467 (2007)
		
		\bibitem{Murphy}
		Murphy, G.\,J.:
		{\em $C^*$-Algebras and Operator Theory}.
		Academic Press (1990)
		
		\bibitem{denittis2025}
		De~Nittis, G.:
		{\sl Topological phases of non-interacting systems: A general approach based on states}.
		\eprint{2502.03590} (2025)
		
		\bibitem{Renault2008}
		Renault, J.:
		{\sl Cartan subalgebras in $C^*$-algebras}.
		Irish Math. Soc. Bull. {\bf 61}, 29--63 (2008)
		
		\bibitem{Archbold}
		Archbold, R.\,J., Bunce, J.\,W., Gregson, K.\,D.:
		{\sl Extensions of states of $C^*$-algebras, II}.
		Proc. R. Soc. Edinburgh A {\bf 92}, no. 1--2, 113--122 (1982)
		
		\bibitem{Anderson}
		Anderson, J.:
		{\sl Extensions, restrictions, and representations of states on $C^*$-algebras}.
		Trans. Amer. Math. Soc. {\bf 249}, 303--329 (1979)
		
		\bibitem{Ferreira:2015uia}
		Ferreira, M.\,J.\,B., Jimenez, J.\,P.\,I., Padmanabhan, P., Sobrinho, P.\,T.:
		{\sl A recipe for constructing frustration-free Hamiltonians with gauge and matter fields in one and two dimensions}.
		J. Phys. A {\bf 48}, no. 48, 485206 (2015); \eprint{1503.07601}
		
		\bibitem{Resende}
		Araujo de Resende, M.\,F.:
		{\sl A pedagogical overview on 2D and 3D toric codes and the origin of their topological orders}.
		Rev. Math. Phys. {\bf 32}, no. 02, 2030002 (2020)
		
		\bibitem{Duivenvoorden}
		Duivenvoorden, K., Breuckmann, N.\,P., Terhal, B.\,M.:
		{\sl Renormalization group decoder for a four-dimensional toric code}.
		IEEE Trans. Inf. Theory {\bf 65}, no. 4, 2545--2562 (2019)
		
		\bibitem{Rotman}
		Rotman, J.\,J.:
		{\em An Introduction to Algebraic Topology}.
		Graduate Texts in Mathematics, vol.~119, Springer-Verlag, New York (1988)
		
		\bibitem{spanier}
		Spanier, E.\,H.:
		{\em Algebraic Topology}.
		Springer, New York (1966)
		
		%\bibitem{Dennis}
		%Dennis, E., Kitaev, A., Landahl, A., Preskill, J.:
		%{\sl Topological quantum memory}.
		%J. Math. Phys. {\bf 43}, no. 9, 4452--4505 (2002)
		
		
		
	\end{thebibliography}
\end{document}